\documentclass[12pt]{article}
\usepackage{amsmath,amsfonts,amsthm,bbm}
\usepackage[shortlabels]{enumitem}
\usepackage{caption,graphicx}
\usepackage{algorithm,algorithmic}
\usepackage{color}
\usepackage{natbib}
\usepackage{textcomp}

\usepackage{ulem} 
\usepackage{url} 

\newcommand{\blind}{0}

\DeclareMathOperator*{\argmin}{arg\,min}
\DeclareMathOperator*{\argmax}{arg\,max}
\newtheorem{theorem}{Theorem}
\newtheorem{lemma}{Lemma}
\newtheorem*{remark}{Remark}

\begin{document}

\def\spacingset#1{\renewcommand{\baselinestretch}%
{#1}\small\normalsize} \spacingset{1}


\if0\blind
{
  \title{\bf Nonparametric Estimation for a Log-concave Distribution Function with Interval-censored Data}
  \author{Chi Wing Chu\\
    Department of Management Sciences, City University of Hong Kong\\
    \\
    Hok Kan Ling \\
    Department of Mathematics and Statistics, Queen's University\\
    \\
    Chaoyu Yuan\\
    Department of Statistics, Columbia University}
  \maketitle
} \fi

\if1\blind
{
  \bigskip
  \bigskip
  \bigskip
  \begin{center}
    {\LARGE\bf Title}
\end{center}
  \medskip
} \fi

\bigskip

\abstract{We consider the nonparametric maximum likelihood estimation for the underlying event time based on mixed-case interval-censored data, under a log-concavity assumption on its distribution function. This generalized framework relaxes the assumptions of a log-concave density function or a concave distribution function considered in the literature. A log-concave distribution function is fulfilled by many common parametric families in survival analysis and also allows for multi-modal and heavy-tailed distributions.
We establish the existence, uniqueness and consistency of the log-concave nonparametric maximum likelihood estimator. A computationally efficient procedure that combines an active set algorithm with the iterative convex minorant algorithm is proposed. Numerical studies demonstrate the advantages of incorporating additional shape constraint compared to the unconstrained nonparametric maximum likelihood estimator. The results also show that our method achieves a balance between efficiency and robustness compared to assuming log-concavity in the density. An R package \texttt{iclogcondist} is developed to implement our proposed method.\\
}

\noindent%
\vfill

\newpage
\spacingset{1.75} 

\section{Introduction}
Interval-censored data arise when a failure event is only known to occur within a specific time interval. This is a prevalent data structure in various clinical trials and longitudinal studies \citep{sun2006statistical} when there is a periodic follow-up. For example, in trials related to acquired immune deficiency syndrome (AIDS), the time from human immunodeficiency virus (HIV) infection to the onset of AIDS cannot be determined exactly, as diagnosis typically relies on periodic blood tests \citep{de1989analysis}. Another example is the emergence of adult teeth \citep{vanobbergen2000tooth}, where the time of emergence is not exactly known because the children only go to a dental checkup annually.

Let $X_1,\ldots,X_N$ be independent random event times with an unknown distribution function $F_0$ on $[0,\infty)$. For interval-censored data, instead of directly observing these event times, we can only inspect each subject $i$ at finitely many time points $0 < T_{i1} < \ldots < T_{i,m(i)} < \infty$, where $m(i)$ denotes the total number of inspections, and determine whether the event has occurred or not. Thus, for each subject, we observe
\begin{equation*}
    Y_{ij} := \mathbbm{1}(T_{i,j-1} < X_i \leq T_{ij}), \quad \text{ for } 1 \leq j \leq m(i) + 1,
\end{equation*}
with $T_{i,0} := 0$ and $T_{i,m(i)+1} := \infty$. Here, $m(i)$ may vary across different subjects. This is referred to as the mixed-case interval-censoring data \citep{schick2000consistency}. A notable special case is when each subject is inspected only once, i.e., $m(i) = 1$ for $i=1,\ldots,N$. In this scenario, the observation $Y_{i1}\equiv 1-Y_{i2}$ represents the event status of subject $i$ at time $T_{i1}$ and is called the current status data. An example of this can be found in cross-sectional studies on survival events \citep{keiding1991age}. More generally, when each subject is inspected exactly $J$ times, it is referred to as case $J$ interval censoring.

The nonparametric maximum likelihood estimator (NPMLE) of the distribution function for interval-censored data has been studied in \cite{turnbull1976empirical}. The NPMLE takes the form of a step function and is known to converge slowly. Under certain regularity conditions, the NPMLE converges at a rate $n^{1/3}$, potentially with an additional logarithmic factor \citep{groeneboom1992information, van1993hellinger}. On the other hand, while parametric models for interval-censored data lead to estimators with a superior convergence rate of $n^{1/2}$, they are vulnerable to model misspecification.

As an alternative, shape-constrained nonparametric estimators offer a valuable option compared to parametric models or other nonparametric methods. Namely, by imposing a qualitative smoothness constraint on the event time distribution, the resulting estimator serves as a middle ground between parametric and fully nonparametric approaches in terms of efficiency and robustness. 
For example, \cite{dumbgen2004consistency} shows that the concave-constrained NPMLE achieves a rate of convergence of $(n/\log (n))^{2/5}$ in the current status model, a substantial improvement over the unconstrained case. Additionally, unlike many other nonparametric approaches, such as bandwidth-dependent kernel methods, shape-constrained nonparametric estimators typically do not require tuning parameters. Shape-constrained estimators are also known to have the potential to adapt to certain characteristics of the unknown function \citep{samworth2018recent}.

In this article, we consider the nonparametric maximum likelihood estimation of the underlying distribution function $F_0$ under the assumption that it is log-concave. A function is said to be log-concave if the logarithm of the function is concave. It is noteworthy that the term ``log-concave distribution" in the literature may refer to the scenario that the density function is log-concave. In contrast, ``log-concave distribution function'' in this paper refers to the distribution function itself being log-concave. In the sequel, we will refer to the NPMLE without any additional shape constraints as the unconstrained MLE, and the NPMLE assuming a log-concave distribution function as the log-concave MLE.

The log-concavity constraint on $F_0$ is appealing because it encompasses a broad class of distributions. Detailed properties and examples of log-concave distribution functions can be found in \cite{sengupta1999log}, \cite{sengupta2004concave}, \cite{bagnoli2005log} and the references therein. Familiar parametric models for survival data, such as the exponential, two-parameter Weibull, log-normal, and log-logistic models \citep{sun2006statistical}, all have log-concave distribution functions. Notably, this class of log-concave distributions contains two interesting subsets:
\begin{enumerate}[(I)]
\item concave distribution functions (or equivalently, monotone decreasing densities): this follows from the fact that logarithm is a concave function;
\item log-concave density functions: see, for example, \cite{bagnoli2005log}.
\end{enumerate}
It is worth outlining some key differences between the class of log-concave distributions and these two subsets. First, a log-concave distribution function can have a multi-modal density, such as the mixture of folded normal distributions, i.e. the mixture of absolute values of normal distributions. In contrast, a log-concave density must be unimodal while a concave distribution function corresponds to monotone decreasing densities. Second, while a log-concave density must have an increasing hazard function \citep{bagnoli2005log}, a log-concave distribution function can exhibit a hazard function that first increases then decreases, as demonstrated by the log-logistic distribution. Third, a log-concave distribution can exhibit a heavy tail, as in the case of the Pareto distribution, while a log-concave density must have an exponentially decaying tail.

For interval-censored data, the estimation of a concave  distribution function was considered in \cite{dumbgen2006estimating}, while the estimation of a log-concave density has been studied by \cite{dumbgen2014maximum} and \cite{anderson2016computing}. The preceding discussion motivates us to consider estimating a distribution function under its log-concavity, which is a more general framework and is broadly applicable in survival analysis. The resulting log-concave MLE does not require any tuning parameters. We establish its existence, uniqueness and consistency following the techniques in \cite{dumbgen2006estimating}. 

To compute the log-concave MLE, we propose an active set algorithm based on the framework of \cite{dumbgen2007active}, coupled with the iterative convex minorant (ICM) algorithm \citep{jongbloed1998iterative}. The active set algorithm involves a constrained optimization with only a small number of variables compared to the total sample size. The constrained optimization can then be solved efficiently via the ICM algorithm. With careful implementation, the computational complexity of each iteration in the ICM algorithm and the computation of the criterion for adding a knot in the active set algorithm is $O(n)$. Our implementation of the algorithm converges in an average of $0.1$ to $0.2$ seconds for a sample size of $1000$, and $2$ to $3$ seconds for a sample size of $10000$, on a laptop computer (CPU: Intel i9-12900HK 2.50 GHz). An R package \texttt{iclogcondist} is developed to implement the proposed algorithm and is available on GitHub\footnote{https://github.com/ChaoyuYuan/iclogcondist}.

The remainder of this paper is organized as follows. In Section \ref{sect:exist_unique_consistency}, we establish the existence, uniqueness, and consistency of the nonparametric maximum likelihood estimator of a log-concave underlying distribution function based on mixed-case interval-censored data. Section \ref{sect:computational_algo} details the proposed computational approach, which combines an active set algorithm with the ICM algorithm. In Section \ref{sect:simulation}, we present simulation studies to demonstrate the effectiveness of both the proposed estimator and the computational algorithm. Section \ref{sect:real_data} illustrates the application of the estimator using two real datasets. We conclude with a discussion in Section \ref{sect:discussion}. Proofs of theoretical results and supplementary simulation studies are included in the appendix.

\section{Estimation under log-concavity}\label{sect:exist_unique_consistency}
We adopt the standard assumptions of independent observations across subjects and an independent censoring mechanism \citep{sun2006statistical}. 
For each $i$, there is a random interval $(L_i, R_i] = (T_{i,J(i)-1}, T_{i,J(i)}]$ containing the event time $X_i$, where $1 \leq J(i) \leq m(i) + 1$. Suppose that $w_i$ observations share the same interval $(L_i, R_i]$ and there are $n$ such distinct intervals. The log-likelihood of the observed data in terms of a distribution function $F$ is given by
\begin{equation*}
    l(F) = \sum_{i=1}^n w_i \log (F(R_i) - F(L_i)),
\end{equation*}
with the convention that $\log(x) := -\infty$ for $x \leq 0$. We are interested in estimating the true distribution function $F_0$ under the constraint that $F_0$ is log-concave. Let $F = e^\psi$ for some concave function $\psi$ such that $\psi(0) = -\infty$ and $\psi(\infty) = 0$. The log-likelihood can be rewritten in terms of $\psi$ as
\begin{equation*}
    l(\psi) = \sum_{i=1}^n w_i \log (e^{\psi(R_i)} - e^{\psi(L_i)}).
\end{equation*}
Under a log-concavity on $F_0$, we maximize $l$ over
\begin{equation*}
   \mathcal{C} := \{\psi:[0,\infty)\rightarrow [-\infty,\infty) \,|\, \psi \text{ is concave},\, \psi(0) = -\infty,\, \psi(\infty)=0\}.
\end{equation*}
Define the set of all distinct time points
\begin{equation*}
    \{0, L_1, R_1,\ldots,L_n, R_n,\infty\} =\{\tau_0,\tau_1,\ldots, \tau_{m+1}\},
\end{equation*}
where $0 =: \tau_0 < \tau_1 <\tau_2 < \ldots < \tau_{m+1} :=\infty$. It is clear that the log-likelihood $l$ depends on $\psi$ only at the points $\tau_i$'s. Thus, the optimization problem requires further specification, 
as the choice of $\psi$ between the $\tau_i$'s to preserve concavity is not unique. To this end, we restrict our focus to maximizing $l$ over the subclass $\mathcal{C}_n \subset \mathcal{C}$, consisting of functions that are linear on each interval $[\tau_{i-1}, \tau_{i}]$, for $i=2,\ldots,m$. This requirement is analogous to the unconstrained estimation of the distribution function for interval censored data, where the unconstrained MLE is conventionally taken as a step function. This type of linear confinement is also common in convex or concave regression; see \cite{groeneboom2014nonparametric} for a detailed exposition on this and other shape-constrained problems.



As a result, we can identify $\psi \in \mathcal{C}_n$ with the $m$-dimensional vector $\phi = (\phi_j)_{j=1}^m := (\psi(\tau_j))_{j=1}^m$. For each $i = 1,\ldots,n$, there are indices $l(i)$ and $r(i)$ such that $    (L_i, R_i] = (\tau_{l(i)}, \tau_{r(i)}]$,
where $l(i)$ can be $0$ and $r(i)$ can be $m+1$. For a generic vector $y \in \mathbb{R}^m$, denote $\Delta y_i := y_i - y_{i-1}$ for $i=1,\ldots,m$ with $y_{0} := 0$.  Then, $l$ can be expressed as
\begin{equation*}
    l(\phi) = \sum_{i=1}^n w_i \log(e^{\phi_{r(i)}} - e^{\phi_{l(i)}}),
\end{equation*}
with $\phi_0 := -\infty$, $\phi_{m+1} = 0$, $\phi \in \mathcal{M} \cap \mathcal{LC}$, where
\begin{align}
    \mathcal{M} &:= \{\phi \in \mathbb{R}^m: \phi_1 \leq \ldots \leq \phi_m \leq 0\}, \label{eq:monotone_constraint} \\
    \mathcal{LC} &:= \left\{ \phi \in \mathbb{R}^m :  
    \frac{\Delta \phi_i}{\Delta \tau_i} \geq \frac{\Delta \phi_{i+1}}{\Delta \tau_{i+1}},
    \quad i=2,\ldots,m-1 \right\} \label{eq:log_concave_constraint},
\end{align}
represent the monotonicity and log-concavity constraints on the corresponding distribution function, respectively. Note that if $\phi \in \mathcal{LC}$, we only need two additional constraints, $\phi_{m-1} \leq \phi_m$ and $\phi_m \leq 0$, such that $\phi$ is also in $\mathcal{M}$; see also Section \ref{sect:computational_algo}.
In Lemma \ref{lemma:log_like_concave}, we show that the log-likelihood function $l(\phi)$ is concave, though not strictly concave. Nevertheless, Theorem \ref{thm:unique_existence} shows that there exists a unique maximizer of $l$ subject to \eqref{eq:monotone_constraint} and \eqref{eq:log_concave_constraint}. Compared with estimating a concave distribution function in \cite{dumbgen2006estimating}, we parameterize the log-likelihood in terms of the log-distribution function, where the maximizer takes the value of negative infinity at some points.
\begin{lemma}\label{lemma:log_like_concave}
The log-likelihood $l(\phi)$ is concave but not strictly concave in $\phi$.
\end{lemma}  
\begin{proof}
    See Appendix \ref{sect:proof}.
\end{proof}

\begin{theorem}\label{thm:unique_existence}
    There exists a unique maximizer  of $l$ subject to (\ref{eq:monotone_constraint}) and (\ref{eq:log_concave_constraint}).
\end{theorem} 
\begin{proof}
    See Appendix \ref{sect:proof}.
\end{proof}


Let $\hat{\phi}_{lc} := \argmax_{\phi \in \mathcal{M} \cap \mathcal{LC}}l(\phi)$. The function estimate of $\log F_0$ is given by 
\begin{equation*}
    \hat{\phi}_{lc}(t) =  \frac{\tau_{i+1}-t}{\tau_{i+1} - \tau_i} \hat{\phi}_{lc}(\tau_i) + 
    \frac{t - \tau_i}{\tau_{i+1} - \tau_i} \hat{\phi}_{lc}(\tau_{i+1}), \quad t \in [\tau_i, \tau_{i+1}],
\end{equation*}
for $i=1,\ldots,m-1$, and $\hat{\phi}_{lc}(t) = -\infty$ for $t \in [0, \tau_1)$. We leave $\hat{\phi}_{lc}(t)$ undetermined for $t > \tau_m$, similar to the unconstrained case \citep{groeneboom1992information}. The corresponding estimator of $F_0$ is defined as $\hat{F}_{lc}(t) := e^{\hat{\phi}_{lc}(t)}$.

To introduce the consistency result in Theorem \ref{thm:consistency} below, we consider a triangular scheme of observations as in \cite{dumbgen2006estimating}, where $m(i) = m^{(n)}(i)$, $T_{ij} = T_{ij}^{(n)}$ and $Y_{ij} = Y_{ij}^{(n)}$. The inspection times are treated as arbitrary fixed numbers. Let $M_n$ denote the empirical distribution of our inspection times, namely
    \begin{equation*}
        M_n(t) := \frac{1}{n}\sum_{i=1}^n \left( \frac{1}{m^{(n)}(i)}\sum_{j=1}^{m^{(n)}(i)} \mathbbm{1}(T_{ij}^{(n)} \leq t)\right).
    \end{equation*}
Let $\mathcal{F}_{lc}$ denotes the set of all log-concave distribution functions on $[0,\infty)$. By Theorem 3 in \cite{dumbgen2006estimating}, we obtain the consistency of the log-concave MLE.
\begin{theorem}\label{thm:consistency}
    Suppose that $F_0 \in \mathcal{F}_{lc}$ and $\frac{1}{n}\sum_{i=1}^n [m^{(n)}(i)]^\gamma = O(1)$ for some $\gamma \in (1/2, 1]$. Then,
    \begin{equation}\label{eq:consistency}
    \int |\hat{F}_{lc} - F_0| \, dM_n \stackrel{p}{\rightarrow} 0.
    \end{equation}
    This implies that, for any $0 \leq s \leq t \leq \infty$ for which $\liminf_{n \rightarrow \infty} (M_n(t) - M_n(s-))> 0$, we have
    \begin{equation*}
        \hat{F}_{lc}(s) \leq F_0(t) + o_p(1) \quad \text{ and }\quad \hat{F}_{lc}(t) \geq F_0(s) + o_p(1).
    \end{equation*}
\end{theorem}
\begin{proof}
    See Appendix \ref{sect:proof}.
\end{proof}

\section{Computational algorithm}\label{sect:computational_algo}
\subsection{Preliminary step}\label{subsect:prelim}

Define the index set
\begin{equation}\label{eq:mathcal_L}
    \mathcal{L} := \left\{i \in\{1,\ldots,n\}: L_i < \min_{j=1,\ldots,n} R_j\right\}.
\end{equation}
Note that $\mathcal{L}$ is not empty since $L_{i^*} < R_{i^*}$, where $i^* := \argmin_{j=1,\ldots,n} R_j$.
For $i \in \mathcal{L}$, the values of $\phi(L_i)$ appear only in expressions of the form $\log(e^{\phi(R_i)} - e^{\phi(L_i)})$ in the log-likelihood, and the $L_i$'s for $i \in \mathcal{L}$ are also among the earliest observed times. It follows that the maximizer $\hat{\phi}$ of $l$ must satisfy the condition $\hat{\phi}(L_i) = -\infty$.  Equivalently, $\hat{\phi}_i = -\infty$ for $i=1,\ldots,s^*$, where $s^* := |\{L_i : i \in \mathcal{L}\}|$. Any value of $\phi(L_i) \neq -\infty$ for $i \in \mathcal{L}$ results in a strictly smaller log-likelihood. 
Accordingly, we redefine the set of distinct time points as:
\begin{equation*}
    \{ \tau_1,\ldots, \tau_m,\tau_{m+1}\} = \{ L_1,R_1,\ldots,L_n,R_n,\infty\} \, \backslash \, \{L_i: i \in \mathcal{L}\},
\end{equation*}
where $\tau_{m+1} = \infty$. For simplicity, we continue to denote the total number of these time points by $m$, reusing the notation from before. Define another index set $\mathcal{R} := \{i :r(i) = m+1\}$. For $i \in \mathcal{R}$, $\phi_{r(i)} = 0$. 
The log-likelihood can then be expressed as:
\begin{align*}
    l(\phi) 
    &= \sum_{i \in \mathcal{L} \cap \mathcal{R}^c} w_i \phi_{r(i)} +
    \sum_{i \in \mathcal{L}^c \cap \mathcal{R}} w_i \log \left( 1 - e^{\phi_{l(i)}}\right) +
    \sum_{i \in \mathcal{L}^c \cap \mathcal{R}^c} w_i \log \left( e^{\phi_{r(i)}} - e^{\phi_{l(i)}}\right),
\end{align*}
since for $i \in \mathcal{L} \cap \mathcal{R}$, we have $\log(e^{\phi(R_i)} - e^{\phi(L_i)}) = 0$.


\subsection{Active set algorithm}
The active set method is an optimization technique used for solving constrained optimization problems by identifying active inequality constraints and treating them as equality constraints. This allows for the solution of a simpler equality-constrained subproblem. An example of its application in statistics is the nonparametric estimation of a log-concave density based on uncensored data \citep{dumbgen2007active}. 

Recall that our log-concave MLE is a piecewise linear function, which can be completely determined by its values at the knots, namely, the points where the slope of the function changes. These knots correspond to the inactive constraints in the optimization problem. In various estimation problems involving concave or log-concave functions, shape-constrained estimators typically have only a small fraction of knots relative to the sample size. For instance, for $n = 1000$, the log-concave MLE based on uncensored standard normally distributed observations has, on average, about 10 knots, while for $n = 100000$, there are approximately 30 knots \citep{wang2018computation}. A similar conclusion holds for our application as we illustrate in Section \ref{sect:simulation}.
Consequently, an active set algorithm is particularly well-suited for such optimization problems, as it only requires solving subproblems where the number of variables corresponds to the number of knots. To this end, we follow the framework of the active set algorithm as described in \cite{dumbgen2007active}.

Let $\text{dom}(l) := \{\phi \in \mathbb{R}^m : l(\phi) > -\infty\}$.
    For $i=2,\ldots,m-1$, let $v_i = (v_{ij})_{j=1}^m$ be a vector in $\mathbb{R}^m$ with exactly three nonzero components:
\begin{equation*}
    v_{i,i-1} := \frac{1}{\Delta \tau_i}, \quad v_{i,i} := - \left(\frac{1}{\Delta \tau_{i+1}} + \frac{1}{\Delta \tau_i} \right), \quad  v_{i,i+1} := \frac{1}{\Delta \tau_{i+1}}.
\end{equation*}
Let $v_m \in \mathbb{R}^m$ with exactly two nonzero components: $v_{m,m-1} = 1$ and $v_{m,m} = -1$. The constraint set $\mathcal{M} \cap \mathcal{LC}$ can be written as $\mathcal{K} \cap \text{dom}(l)$, where
\begin{equation*}
    \mathcal{K} := \{ \phi \in \mathbb{R}^m: v_i^\top \phi \leq 0,\, i=2,\ldots,m\}.
\end{equation*}
Note that there are only $m-1$ constraints and the index in $\mathcal{K}$ begins at $i = 2$ is for convenience, as $v_i^\top \phi = 0$ corresponds to the scenario where $\phi_{\tau_i}$ is not a knot for $i=2,\ldots,m$. Thus, the indices for the constraints in $\mathcal{K}$ are $2,\ldots,m$.
 For any $\phi \in \mathbb{R}^m$, define the index set
\begin{equation*}
    A(\phi) := \{ i \in\{2,\ldots,m\}: v_i^\top \phi \geq 0\}.
\end{equation*}
The following theorem, based on standard convex optimization theory, provides a characterization of the maximizer using a set of basis vectors, which will be used to determine which knots to be added as well as the stopping criterion in the algorithm.
\begin{theorem}[Theorem 3.1 in \cite{dumbgen2007active}]\label{thm:characterization}
    Let $b_1,\ldots,b_m$ be a set of basis vectors of $\mathbb{R}^m$ such that $v_i^\top b_j < 0$ if $i = j$ and $v_i^\top b_j = 0$ for $i\neq j$. A vector $\phi \in \mathcal{K} \cap \text{dom}(l)$ is the maximizer of $l$ if and only if 
    \begin{equation}\label{eq:characterizatoin}
    b_i^\top \nabla l(\phi) 
    \begin{cases}
    =0, & \text{ for all } i \in \{1,\ldots,m\} \, \backslash \, A(\phi); \\
    \leq 0, & \text{ for all } i \in A(\phi).
    \end{cases}
    \end{equation}
\end{theorem}
A suitable set of basis vectors of $\mathbb{R}^m$ that satisfies the conditions in Theorem \ref{thm:characterization} is given by $b_1 := (1)_{i=1}^m$, and for $j = 2,\ldots,m$,
\begin{equation*}
    b_j := (\min(\tau_i - \tau_j, 0))_{i=1}^m.
\end{equation*}
It can be readily checked that $v_i^\top b_j = 0$ for $i \neq j$ and $v_i^\top b_i < 0$ for $i=2,\ldots,m$. 
\begin{remark}
    Compared to estimating a log-concave density with uncensored data, we have one additional constraint, $\phi_{m-1} \leq \phi_m$. Thus, the set of basis vectors is also modified accordingly.
\end{remark}
For $A \subseteq \{2,\ldots,m\}$, let 
\begin{equation*}
    \mathcal{V}(A) := \{\phi \in \mathbb{R}^m : v_a^\top \phi = 0  \text{ for all } a \in A\}.
\end{equation*}
The active set algorithm performs two basic procedures alternatively. The first procedure is to replace a feasible point $\phi \in \mathcal{K} \cap \text{dom}(l)$ with a conditionally optimal one. Set $A = A(\phi)$. We first solve
\begin{equation*}
\phi_{\text{cand}} = \tilde{\phi}(A) := \argmax_{\phi \in \mathcal{V}(A)\cap \mathcal{M}} l(\phi),
\end{equation*}
using the ICM algorithm as described in Subsection \ref{subsect:solving_subproblem}. As $\phi_{\text{cand}}$ may not be in $\mathcal{K}$, we replace $\phi$ by $(1-t) \phi + t \phi_{\text{cand}}$, where 
\begin{equation*}
    t = t(\phi, \phi_{\text{cand}}) := \max\{t \in (0, 1): (1-t)\phi+t \phi_{\text{cand}} \in \mathcal{K}\}.
\end{equation*}
Then, we repeat the procedure until $\phi_\text{cand} \in \mathcal{K}$. See Procedure \ref{algo:adjust} for the algorithmic steps. The second procedure is to alter the set of active constraints by adding a knot $a_0$ that corresponds the largest positive $b_{a_0}^\top \nabla l(\phi)$, based on the characterization in (\ref{eq:characterizatoin}). The active set algorithm iterates these two procedures until a convergence criterion is met; see Algorithm \ref{algo:full}.

\captionsetup[algorithm]{name=Procedure}  
\begin{algorithm}[H]
\caption{Pseudo-code of finding a conditionally optimal one}
\label{algo:adjust}
\begin{algorithmic}
\STATE $\phi_{\text{cand}} \gets \tilde{\phi}(A)$
\WHILE{$\phi_{\text{cand}} \notin \mathcal{K}$}
        \STATE $\phi \gets (1 - t(\phi, \phi_{\text{cand}})) \phi + t(\phi, \phi_{\text{cand}}) \phi_{\text{cand}}$
        \STATE $A \gets A(\phi)$
        \STATE $\phi_{\text{cand}} \gets \tilde{\phi}(A)$
    \ENDWHILE
        \STATE $\phi \gets \phi_{\text{cand}}$
    \STATE $A \gets A(\phi)$
\end{algorithmic}
\end{algorithm}
\captionsetup[algorithm]{name=Algorithm}  

\begin{algorithm}[H]
\caption{Pseudo-code of an active set algorithm for computing $\hat{\phi}_{lc}$}
\label{algo:full}
\begin{algorithmic}
\STATE $\phi \gets \phi_0$
\STATE $A \gets A(\phi)$

\STATE Apply Procedure \ref{algo:adjust}
\WHILE{$\max_{a \in A} b_a^T \nabla l(\phi) > 0$}
    \STATE $a \gets \min \left(\arg\max_{a \in A} b_a^T \nabla l(\phi)\right)$
    \STATE $A \gets A \setminus \{a\}$

\STATE Apply Procedure \ref{algo:adjust}
\ENDWHILE
\end{algorithmic}
\end{algorithm}

We now provide the details of applying this active set algorithm to our problem. First, an initial value $\phi_0$ is needed. A natural choice is to find the unconstrained MLE and evaluate the least concave majorant of its logarithm at $\tau_i$'s. This vector lies in $\mathcal{K}\cap \text{dom}(l)$, typically has only a few knots, and is not too far away from the optimal value. While the unconstrained MLE itself could also be an initial value, it generally requires more iterations to adjust the candidate solution $\phi_{\text{cand}}$ to lie within $\mathcal{K}$ in the first time of Procedure 1. The unconstrained MLE can be efficiently computed using the EMICM algorithm implemented in \cite{anderson2017icenreg}. The computation of the least concave majorant can be done using the pool adjacent violators algorithm (PAVA) in $O(n)$ time \citep{busing2022monotone}. 

The detail of computing $\tilde{\phi}(A)$ is given in Subsection \ref{subsect:solving_subproblem}. In the active set algorithm, we have to compute $\tilde{\phi}(A)$ multiple times for different $A$. The initial value for the ICM step is taken to be the current value of $\phi$.

In Procedure 1, $t(\phi, \phi_{\text{cand}})$ can be computed as
\begin{equation*}
    t(\phi, \phi_{\text{cand}}) = \min \left\{ \frac{-v_i^\top \phi}{v_i^\top \phi_{\text{cand}} - v_i^\top \phi} : 2 \leq i \leq m-1, v_i^\top \phi_{\text{cand}} > 0\right\}.
\end{equation*}
This does not involve the $m$th constraint because both $\phi$ and $\phi_{\text{cand}}$ always satisfy $v_m^\top \phi \leq 0$ and $v_m^\top \phi_{\text{cand}} \leq 0$.

The active set algorithm determines which knot to add and when to stop by computing $b_a^\top \nabla L(\phi)$ for $a \in A$, $A \in \{2,\ldots,m\}$. A naive implementation requires $O(n^2)$ computation as $m = O(n)$. To compute  $b_a^\top \nabla L(\phi)$ efficiently for all $a \in A$, we note that for any vector $x \in \mathbb{R}^m$, 
\begin{align*}
    b_2^\top x &= (\tau_1 - \tau_2) x_1, \\
    b_j^\top x &= b_{j-1}^\top x + (\tau_{j-1} - \tau_j) \sum_{i=1}^{j-1}x_i,
\end{align*}
for $j=3,\ldots,m$. Based on this recursive relation, we can compute $b_a^\top \nabla l(\phi)$  for all $a \in A$ in $O(n)$ computations. 

To decide the convergence of the algorithm, we can set a small tolerance level $\eta > 0$ such that if $\max_{a \in A} b_a^\top \nabla l(\phi) < \eta$, we stop. Alternatively, we can also stop when the change of log-likelihood between consecutive iterations is smaller than $\varepsilon$. In practice, both criteria will give nearly the same results.

\subsection{Details of solving $\tilde{\phi}(A)$}\label{subsect:solving_subproblem}
Let $A \subset\{2,\ldots,m\}$ and $\phi \in \mathcal{V}(A)$. Recall that the inactive constraints, $v_a^\top \phi < 0 $ for $a \notin A$, correspond to the knots of $\phi$, which is indexed by  $I := \{1,\ldots,m\} \,\backslash \, A$.
Write $I = \{i(1),\ldots,i(k)\}$, where $1=i(1) < \dots < i(k)$.
The values of $\phi(\tau_i)$ for $i \notin I$ are determined by the values of $\phi(\tau_i)$ for $i \in I$. Namely, for $\tau_i \in [\tau_{i(s)}, \tau_{i(s+1)}]$, we have
\begin{equation*}
    \phi(\tau_i) = q_{s,i} \phi(\tau_{i(s)}) + (1-q_{s,i}) \phi(\tau_{i(s+1)}),
\end{equation*}
where for $s=1,\ldots,k-1, i=i(s),\ldots,i(s+1)$,
\begin{equation*}
    q_{s,i} := \frac{\tau_{i(s+1)} - \tau_i}{\tau_{i(s+1)} - \tau_{i(s)}}.
\end{equation*}
For $\tau_i > \tau_{i(k)}$, we have
\begin{equation*}
    \phi(\tau_i) = \phi(\tau_{i(k)}).
\end{equation*}

Let $\bar{\phi} := (\phi_{i(1)},\ldots,\phi_{i(k)})$ and $l_{\text{red}}$ be the log-likelihood in terms of the reduced set such that $l_{\text{red}}(\bar{\phi}) = l(\phi)$. Now, we maximize $l_{\text{red}}(\bar{\phi})$ subject to the constraint $\phi_{i(1)} \leq \ldots \leq \phi_{i(k)} \leq 0$ using the (modified) ICM algorithm \citep{jongbloed1998iterative}. In the ICM algorithm, the negative log-likelihood is approximated by a quadratic approximation where the Hessian matrix of the negative log-likelihood is replaced by a diagonal positive definite matrix $D = (d_s)_{s=1}^k$. The following function is then minimized using the PAVA:
\begin{equation*}
    \sum_{s=1}^k d_s \left( \phi_{i(s)} - \left(\phi_{i(s)}^o + d_s^{-1} \frac{\partial l_{\text{red}}(\bar{\phi}^o)}{\partial \phi_{i(s)}} \right) \right)^2
\end{equation*}
subject to $\phi_{i(1)} \leq \ldots \phi_{i(k)} \leq 0$, where $\bar{\phi}^o$ is the value of the previous iterate. The ICM algorithm iterates the quadratic approximation and PAVA until convergence. Backtracking is used to ensure the decrease in negative likelihood and hence convergence of the algorithm \citep{jongbloed1998iterative}.
For $s=1,\ldots,k$, we have
\begin{align*}
    \frac{\partial l_{\text{red}}(\bar{\phi})}{\partial \phi_{i(s)}} &= \sum_{j=1}^m \frac{\partial l(\phi)}{\partial \phi_j} \cdot \frac{\partial \phi_j}{\partial \phi_{i(s)}}, \\
      \frac{\partial^2 l_{\text{red}}(\bar{\phi})}{\partial \phi_{i(s)}^2} & 
      %
      = \sum_{j=1}^m  \sum_{j'=1}^m \frac{\partial^2 l(\phi)}{\partial \phi_j \partial \phi_{j'}} \cdot \frac{\partial \phi_j}{\partial \phi_{i(s)}}  \cdot \frac{\partial \phi_{j'}}{\partial \phi_{i(s)}},
\end{align*}
where the last equality follows as $\frac{\partial^2 \phi_j}{\partial \phi_{i(s)}^2} = 0$ because $\phi_j$ is linear in $\phi_{i(s)}$.
The first order partial derivatives are
\begin{align*}
\frac{\partial l_{\text{red}}(\bar{\phi})}{\partial \phi_{i(1)}} &
=  \sum_{j=i(1)}^{i(2)-1} \frac{\partial l(\phi)}{\partial \phi_j}
q_{1,j},
\\
    \frac{\partial l_{\text{red}}(\bar{\phi})}{\partial \phi_{i(s)}} &
    = 
    \sum_{j=i(s-1)+1}^{i(s)-1} \frac{\partial l(\phi)}{\partial \phi_j}(1-q_{s-1,j}) + \sum_{j=i(s)}^{i(s+1)-1} \frac{\partial l(\phi)}{\partial \phi_j} q_{s,j}
    , \quad s=2,\ldots,k-1, \\
    \frac{\partial l_{\text{red}}(\bar{\phi})}{\partial \phi_{i(k)}} &
    = 
    \sum_{j=i(k-1)+1}^{i(k)} \frac{\partial l(\phi)}{\partial \phi_j} \cdot (1-q_{k-1,j}) + \sum_{j = i(k)+1}^m \frac{\partial l(\phi)}{\partial \phi_j}.
\end{align*}
If $i(k) = m$, then the summation $\sum_{j=i(k)+1}^m$ has no terms and is equal to $0$.
Note that the majority of the mixed partial derivatives $\frac{\partial^2 l(\phi)}{\partial \phi_j \partial \phi_{j'}}$ for $j\neq j'$ is $0$. Therefore, we define $d_1,\dots,d_k$ as $-\frac{\partial^2 l_{\text{red}}(\bar{\phi})}{\partial \phi_{i(s)}^2}$ without the mixed partial derivatives:
\begin{equation*}
    d_s := -\sum_{j=1}^m \frac{\partial^2 l(\phi)}{\partial \phi_j^2 } \left(\frac{\partial \phi_j}{\partial \phi_{i(s)}} \right)^2.
\end{equation*}
They simplify into
\begin{align*}
    d_1 &= -\sum_{j=i(1)}^{i(2)-1} \frac{\partial^2 l(\phi)}{\partial \phi_j^2} q_{1,j}^2, \\
    d_s &= -\sum_{j=i(s-1)+1}^{i(s)-1} \frac{\partial^2 l(\phi)}{\partial \phi_j^2 }(1-q_{s-1,j})^2 - \sum_{j=i(s)}^{i(s+1)-1} \frac{\partial^2 l(\phi)}{\partial \phi_j^2} q_{s,j}^2, \quad s=2,\ldots,k-1, \\
    d_k &= -\sum_{j=i(k-1)+1}^{i(k)} \frac{\partial^2 l(\phi)}{\partial \phi_j^2}  (1-q_{k,j})^2 - \sum_{j=i(k)+1}^m  \frac{\partial^2 l(\phi)}{\partial \phi_j^2}.
\end{align*}

Noting that $r(i) > l(i)$, if $r(i) = j$, we have $l(i)\neq j$ and vice versa, if $l(i) = j$, we have $r(i)\neq j$. Hence, the first and second derivatives of $l(\phi)$ with respect to $\phi$ can be computed as follows:
\begin{align*}
    \frac{\partial l(\phi)}{\partial \phi_j} 
    &= \sum_{i \in \mathcal{L}\cap \mathcal{R}^c:r(i)=j} w_i -  \sum_{i \in \mathcal{L}^c \cap \mathcal{R}: l(i)=j} \frac{w_i}{e^{-\phi_j}-1} \\
    & \quad \quad  + \sum_{i \in \mathcal{L}^c \cap \mathcal{R}^c: r(i)=j} \frac{w_i}{1 - e^{\phi_{l(i)} -\phi_j}} - 
    \sum_{i \in \mathcal{L}^c \cap \mathcal{R}^c: l(i)=j} \frac{w_i}{e^{\phi_{r(i)} -\phi_j} -1},\\
    \frac{\partial^2 l(\phi)}{\partial \phi_j^2} &= 
     -  \sum_{i \in \mathcal{L}^c \cap \mathcal{R}: l(i)=j} \frac{w_ie^{-\phi_j}}{(e^{-\phi_j}-1)^2} - \sum_{i \in \mathcal{L}^c \cap \mathcal{R}^c: r(i)=j} \frac{w_i e^{\phi_{l(i)} - \phi_j}}{(1-e^{\phi_{l(i)}-\phi_j})^2} 
\\
& \quad \quad  -  \sum_{i \in \mathcal{L}^c \cap \mathcal{R}^c: l(i)=j} \frac{w_i e^{\phi_{r(i)}-\phi_j}}{(e^{\phi_{r(i)}-\phi_j}-1)^2}.
\end{align*}
The key remark for efficient calculation of these derivatives is that each observation can at most contribute to two terms in the gradient or the Hessian matrix. Thus, calculating all the terms in the gradient and the Hessian matrix requires $O(n)$ computation.

\section{Simulation studies}\label{sect:simulation}
In this section, we present some simulation results of the log-concave MLE and a couple of contenders. Recall that $\hat{F}_{lc}$ denotes the log-concave MLE of the underlying distribution function $F_0$, assuming the distribution function is log-concave. We set $\hat{F}_{lc}(t) = \hat{F}_{lc}(\tau_m)$ for $t > \tau_m$. Let $\hat{F}_{un}$ be the unconstrained MLE of $F_0$, without imposing any additional shape constraints. Also, $\hat{F}_{lcd}$ represents the estimate of the distribution function derived from the MLE of a log-concave density, assuming the underlying density function $f_0$ is log-concave. Computation of $\hat{F}_{un}$ and $\hat{F}_{lcd}$  are done using the \verb|R| package \verb|icenReg| \citep{anderson2017icenreg} and
\verb|logconPH| \citep{Anderson-Bergman2014logconPH}, respectively.

\subsection{Visual comparison}
In Figure \ref{fig:weibull}, we compare $\hat{F}_{lc}, \hat{F}_{un}$ and $\hat{F}_{lcd}$ visually based on one particular dataset with a sample size of $500$. We consider a case 2 interval censoring with censoring times $C_1$ and $C_2$ generated from Unif$(0, 1)$ and Unif$(C_1, 2)$ respectively. The event times $T$ are generated from a truncated Weibull distribution on $[0, 2]$ with a scale parameter of $1$ and varying shape parameters. Note that the density function for a Weibull distribution is log-convex when the shape parameter is less than $1$; it is log-linear and coincides with the exponential distribution when the shape parameter equals $1$; and it is log-concave when the shape parameter is greater than $1$. Nevertheless, the distribution function remains log-concave in all scenarios. Figure \ref{fig:weibull} shows that when the underlying density function $f_0$ is log-concave, both $\hat{F}_{lc}$ (blue dashed line) and $\hat{F}_{lcd}$ (purple dot-dashed line) are very close to the underlying distribution function $F_0$ (red solid line). However, when the log-concavity only holds for $F_0$ but not $f_0$, $\hat{F}_{lcd}$ exhibits a substantial bias because the model assumption is violated. The nonparametric MLE $\hat{F}_{un}$ (black dotted line) also centers around $F_0$, but it is a nonsmooth step function.
\begin{figure}[ht]
    \captionsetup{font=footnotesize}
    \centering
    \includegraphics[width=1\linewidth]{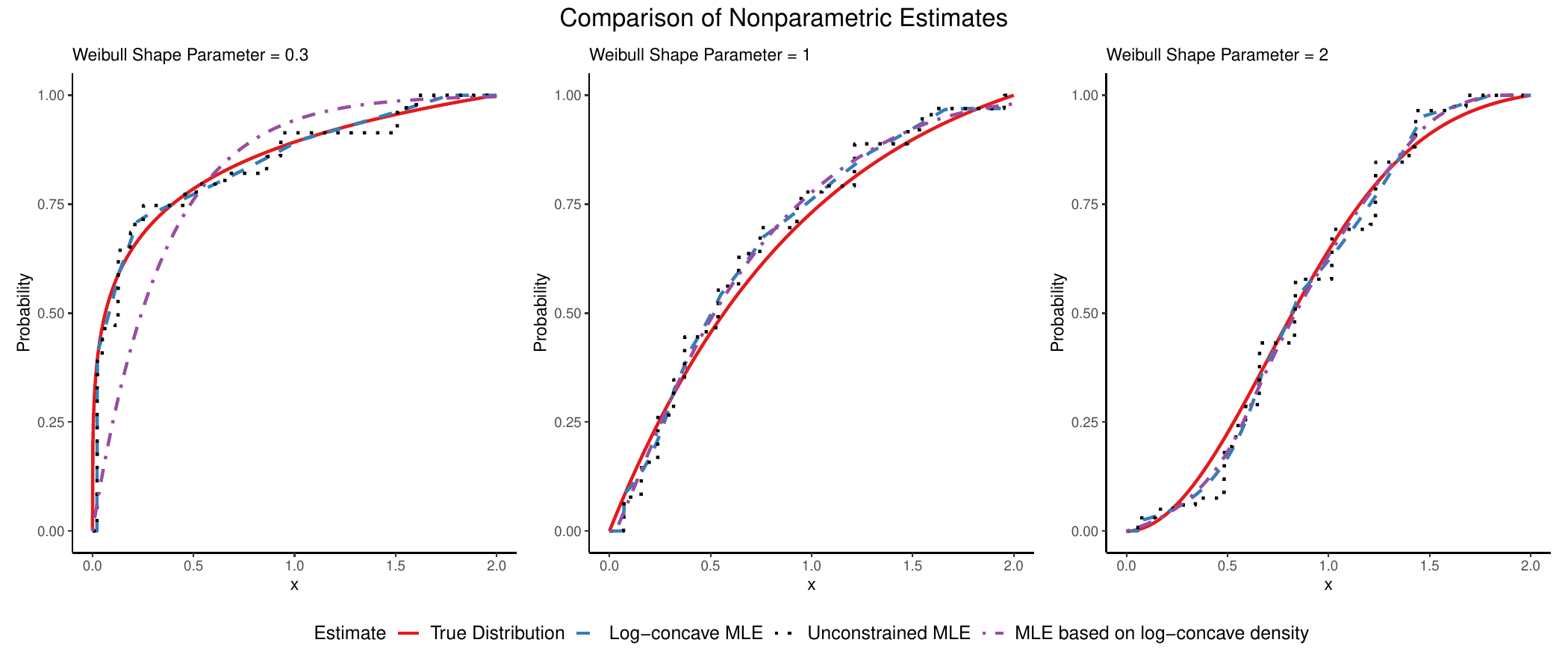}
    \caption{Nonparametric Estimates for an Interval-censored Weibull distributed Dataset with $N=500$}
    \label{fig:weibull}
\end{figure}

\subsection{Computational performance}
Next, we evaluate the performance of our computational algorithm and its implementation. The algorithm is terminated if either of these convergence criteria holds: $\max_{a \in A} b_a^\top \nabla l(\phi) < \eta$ or the change of log-likelihood is less than $\eta$, where $\eta = 10^{-10}$. We follow the same setting described above for obtaining Figure \ref{fig:weibull}, with $N=1000$ and $N=10000$. Because both the event times and censoring times are continuously distributed, there are no ties in the data and thus the weights $w_i\equiv 1$. Table \ref{table:time} shows the summary statistics over $100$ replications. $t_\text{mean}$ and $t_{\text{median}}$ denote the average and median times in seconds for the algorithm to converge, while $t_{\text{mean}}^{lcd}$ and $t_{\text{median}}^{lcd}$ are the corresponding durations for computing the MLE based on a log-concave density assumption. \# of knots indicates the average number of knots for the log-concave MLE, while the value in parentheses represents the corresponding standard deviation. $L_1(\hat{F}_{lc})$ denotes the average over $100$ replicates of $\frac{1}{M}\sum_{i=1}^M |\hat{F}_{lc}(y_i) - F_0(y_i)|$, each calculated over $M=1000$ uniformly spaced grid points $y_i$ on $[0, 2]$. We similarly define the empirical absolute bias as $L_1(\hat{F}_{un})$ and $L_1(\hat{F}_{lcd})$, corresponding to the unconstrained MLE and the MLE under the log-concave density assumption, respectively. The proposed algorithm runs reasonably fast: converging within $0.2$ seconds for $N = 1000$ and around $2$ to $3$ seconds for $N = 10000$. On the other hand, the computation of $\hat{F}_{lcd}$ takes considerably more time when the underlying density is log-concave (shape parameter $\geq 1$). All the computations were performed on a laptop equipped with an intel i9-12900HK CPU at 2.50 GHz and 64 GB RAM. 

\begin{table}[ht]
\captionsetup{font=footnotesize}
\resizebox{\textwidth}{!}{
\centering
\begin{tabular}{rrrrrrrrrr}
  \hline
 $N$ & shape & $t_{\text{mean}}$ & $t_{\text{median}}$  & \# of knots& $t_{\text{mean}}^{lcd}$ & $t_{\text{median}}^{lcd}$ &  $L_1(\hat{F}_{lc}$) & $L_1(\hat{F}_{un}$) & $L_1(\hat{F}_{lcd})$\\ 
  \hline
 1000& 0.3 & 0.17 & 0.11 & 8.38  (1.38) & 0.08 & 0.08 & 1.25 & 1.88 & 6.71\\ 
   & 1& 0.18 & 0.13 & 9.01 (1.62) & 13.8 & 13.4 & 1.56 & 2.49 & 1.20 \\ 
   &2 & 0.14 & 0.12 & 8.77  (1.36) & 12.6 & 12.5 & 1.44 & 2.38 & 1.24\\ 
  10000 & 0.3  & 2.81 & 1.97 & 13.95 (1.76) &3.12& 3.09  & 0.45 & 0.80 & 6.69 \\ 
&  1 & 3.12 & 2.10 & 14.90  (2.01)& 83.6 & 40.3  & 0.57 & 1.08 & 0.90\\ 
   &2 & 2.54 & 2.14 & 15.29  (1.87) & 1044 & 941 & 0.56 & 1.05 & 0.49\\ 
   \hline
\end{tabular} }
\caption{Computational Performance of the Proposed Algorithm for the Log-concave MLE. The empirical absolute bias $L_1(\hat{F}_{lc})$, $L_1(\hat{F}_{un})$ and $L_1(\hat{F}_{lcd})$ are reported in the units of $10^2$. The underlying distribution of the event time is Weibull, with the shape parameter specified in the table and scale parameter equals $1$.}
\label{table:time}
\end{table}

\begin{table}[ht]
\captionsetup{font=footnotesize}
\resizebox{\textwidth}{!}{
\centering
\begin{tabular}{r|c|rrrrr||rrrrr}
  \hline
   & & \multicolumn{5}{c||}{Exp(1)} & \multicolumn{5}{c}{Weibull(2, 1)} \\  \hline
 $N$ &Est. & 0.1 & 0.3 & 0.5 & 0.7 & 0.9 & 0.1 & 0.3 & 0.5 & 0.7 & 0.9 \\ 
  \hline
250 & $\hat{F}_{lc}$  & -1.84 (6.77) & 0.03 (4.79) & 0.13 (4.3) & 0.3 (3.79) & 0.51 (2.83) & 0.17 (3.99) & 0.01 (4.19) & 0.13 (4.28) & -0.12 (4.49) & 0.48 (3.44) \\ 
 & $\hat{F}_{un}$ & -3.12 (6.29) & -0.52 (7.31) & 0.06 (6.68) & 0.4 (5.5) & 0.16 (4.66) & -1.17 (4.81) & -0.41 (6.88) & 0.15 (6.98) & -1.11 (7.56) & 0.63 (4.97) \\ 
  & $\hat{F}_{lcd}$ & -2.1 (2.55) & -1.58 (2.98) & -0.4 (3.29) & 0.83 (3.06) & 1.13 (2.05) & -0.15 (3.02) & 0.22 (3.48) & -0.07 (3.72) & -0.3 (3.78) & 0.25 (2.74) \\  \hline
500  & $\hat{F}_{lc}$  & -0.26 (4.53) & 0.14 (3.59) & 0.04 (3.09) & 0.15 (2.87) & 0.19 (2.03) & 0.12 (2.5) & 0.04 (3.15) & 0.28 (3.23) & 0.16 (3.51) & 0.1 (2.52) \\ 
  &$\hat{F}_{un}$ & -1.64 (4.9) & -0.07 (5.61) & -0.04 (4.96) & 0.3 (4.16) & 0.07 (3.36) & -0.6 (3.87) & -0.27 (5.12) & 0.06 (5.4) & -0.28 (5.82) & 0.24 (3.72) \\ 
  &$\hat{F}_{lcd}$ & -1.47 (1.73) & -1.16 (2.1) & -0.34 (2.37) & 0.63 (2.22) & 0.82 (1.43) & -0.09 (2.13) & 0.19 (2.61) & 0.15 (2.73) & 0 (2.74) & 0.15 (2.02) \\ \hline
1000  &$\hat{F}_{lc}$  & 0.12 (2.9) & 0.04 (2.71) & 0.07 (2.34) & 0.19 (2.07) & 0.09 (1.56) & 0.1 (1.8) & 0 (2.45) & 0.17 (2.45) & -0.06 (2.59) & 0.07 (1.93) \\ 
  &$\hat{F}_{un}$ & -0.89 (3.91) & -0.17 (4.28) & 0.13 (3.83) & 0.28 (3.14) & 0.1 (2.67) & -0.17 (2.95) & 0 (4.23) & 0.18 (4.25) & -0.6 (4.39) & 0.3 (2.95) \\ 
  &$\hat{F}_{lcd}$ & -1.05 (1.19) & -0.85 (1.51) & -0.22 (1.68) & 0.6 (1.64) & 0.69 (1.03) & 0.02 (1.53) & 0.06 (1.93) & 0.06 (2.02) & -0.07 (2) & -0.02 (1.41) \\ 
   \hline
\end{tabular}
}
\caption{The estimated bias values are reported in units of $ \times 10^2$, and the numbers in parentheses indicate the estimated standard deviations, also in units of $ \times 10^2$. Exp(1) denotes the truncated exponential distribution with scale parameter $1$ on $[0, 2]$. Weibull(2, 1) denotes the truncated Weibull distribution with shape parameter $2$ and scale parameter $1$ on $[0, 2]$.}
\label{table:weibull_1_2}
\end{table}

\begin{table}[ht]
\captionsetup{font=footnotesize}
\resizebox{\textwidth}{!}{
\centering
\begin{tabular}{r|c|rrrrr||rrrrr}
  \hline
  & & \multicolumn{5}{c||}{Log-logistic(0.5, 1) } & \multicolumn{5}{c}{Log-logistic(1, 1)} \\ \hline
 $N$ & Est.& 0.1 & 0.3 & 0.5 & 0.7 & 0.9 & 0.1 & 0.3 & 0.5 & 0.7 & 0.9 \\ 
  \hline
$250$ & $\hat{F}_{lc}$ & -7.41 (7.31) & 0.21 (7.34) & -0.01 (4.47) & -0.19 (4.54) & 0.24 (2.87) & -0.92 (6.47) & 0.05 (4.57) & 0.14 (4.58) & -0.7 (5.69) & 0.07 (3.18) \\ 
 & $\hat{F}_{un}$ & -7.73 (6.46) & -2.42 (9.16) & -0.15 (6.32) & -0.63 (6.49) & 0.08 (4.02) & -2.53 (5.83) & -0.37 (6.65) & 0.54 (6.13) & -2.28 (8.38) & 0.04 (4.35) \\ 
  & $\hat{F}_{lcd}$ & -9.4 (0.08) & -22.49 (0.94) & -21.51 (3.1) & -1.81 (4.66) & 8.27 (0.89) & -3.61 (0.67) & -8.06 (2.07) & -7.72 (3.37) & -1.63 (3.83) & 4.73 (1.58) \\  \hline
  $500$ & $\hat{F}_{lc}$ & -6.54 (7.32) & -0.03 (5.14) & 0.08 (3.33) & -0.19 (3.46) & 0.03 (2.04) & -0.2 (4.5) & 0.08 (3.44) & 0.24 (3.41) & -0.22 (4.45) & 0.01 (2.3) \\ 
  & $\hat{F}_{un}$ & -6.96 (6.51) & -1.33 (6.89) & -0.15 (4.71) & -0.86 (4.99) & 0.05 (2.99) & -1.76 (4.81) & -0.16 (5) & 0.3 (4.7) & -1.28 (6.63) & 0.03 (3.29) \\ 
  & $\hat{F}_{lcd}$ & -9.4 (0.05) & -22.56 (0.65) & -21.71 (2.17) & -1.98 (3.3) & 8.32 (0.62) & -3.65 (0.46) & -8.19 (1.44) & -7.91 (2.37) & -1.8 (2.71) & 4.73 (1.13) \\ \hline
  $1000$ & $\hat{F}_{lc}$ & -5.05 (7.42) & 0.03 (3.86) & -0.05 (2.44) & -0.26 (2.58) & -0.03 (1.6) & 0.25 (2.64) & 0.04 (2.58) & 0.06 (2.39) & -0.17 (3.41) & -0.19 (1.68) \\ 
  & $\hat{F}_{un}$ & -5.89 (6.32) & -0.77 (5.32) & -0.22 (3.7) & -1.08 (3.9) & -0.01 (2.41) & -0.64 (3.57) & -0.13 (3.98) & 0.25 (3.44) & -0.63 (5.17) & -0.07 (2.55) \\ 
  & $\hat{F}_{lcd}$ & -9.4 (0.04) & -22.59 (0.46) & -21.8 (1.53) & -2.07 (2.33) & 8.34 (0.43) & -3.63 (0.31) & -8.13 (0.98) & -7.79 (1.6) & -1.63 (1.83) & 4.83 (0.75) \\ 
   \hline
\end{tabular}
}
\caption{The estimated bias values are reported in units of $ \times 10^2$, and the numbers in parentheses indicate the estimated standard deviations, also in units of $ \times 10^2$. Log-logistic($a$, $b$) denotes the truncated log-logistic distribution on $[0, 10]$ with shape parameter $a$ and scale parameter $b$.}
\label{table:loglogistic_0.5_1}
\end{table}

\begin{table}[ht]
\captionsetup{font=footnotesize}
\resizebox{\textwidth}{!}{
\centering
\begin{tabular}{r|c|rrrrr||rrrrr}
  \hline
  & & \multicolumn{5}{c||}{Log-normal(0, 1) } & \multicolumn{5}{c}{Log-normal(0, 2)} \\ \hline
 $N$ & Est. & 0.1 & 0.3 & 0.5 & 0.7 & 0.9 & 0.1 & 0.3 & 0.5 & 0.7 & 0.9 \\ 
  \hline
$250$ & $\hat{F}_{lc}$ & 0.24 (4.91) & 0.06 (4.64) & 0.18 (5.86) & -0.37 (7.24) & -0.05 (3.66) & -2.55 (7.59) & -0.21 (5.09) & 0.14 (4.19) & -0.66 (5.34) & 0.24 (3.05) \\ 
 & $\hat{F}_{un}$ & -1.38 (5.16) & -0.33 (6.73) & 3.46 (8.55) & -1.68 (10.17) & 0.15 (4.99) & -3.77 (6.74) & -1.18 (7.14) & 0.13 (5.8) & -2.16 (7.98) & 0.32 (4.25) \\ 
  & $\hat{F}_{lcd}$ & 0.57 (2.68) & -1.82 (2.73) & -3.8 (3.34) & -3.19 (3.81) & 0.79 (2.48) & -5.59 (0.5) & -12.64 (1.81) & -10.92 (3.44) & -0.19 (4.06) & 6.34 (1.31) \\  \hline
  $500$ & $\hat{F}_{lc}$ & 0.33 (2.96) & 0.07 (3.25) & 0.16 (4.4) & -0.33 (5.38) & -0.05 (2.76) & -0.55 (5.66) & 0 (3.84) & 0.06 (2.92) & -0.3 (4.02) & 0.02 (2.15) \\ 
  & $\hat{F}_{un}$  & -0.79 (4.13) & -0.08 (5.05) & 2.72 (6.6) & -1.14 (7.85) & 0.07 (3.74) & -2.25 (5.35) & -0.48 (5.34) & 0.04 (4.32) & -1.57 (6.34) & 0.15 (3.12) \\ 
  & $\hat{F}_{lcd}$  & 0.85 (1.79) & -1.59 (1.81) & -3.76 (2.26) & -3.34 (2.59) & 0.57 (1.68) & -5.62 (0.35) & -12.74 (1.28) & -11.1 (2.45) & -0.32 (2.92) & 6.37 (0.94) \\ \hline
  $1000$ & $\hat{F}_{lc}$ & -0.04 (1.98) & -0.08 (2.47) & 0.17 (3.21) & -0.11 (4.05) & -0.02 (2.13) & 0.23 (3.75) & -0.04 (2.71) & 0.06 (2.23) & -0.34 (3.11) & -0.03 (1.62) \\ 
  &$\hat{F}_{un}$  & -0.64 (2.96) & -0.35 (4.02) & 1.74 (4.9) & -0.36 (6.2) & 0.01 (3.01) & -1.17 (4.24) & -0.48 (4.17) & 0.12 (3.31) & -1.21 (4.75) & 0.07 (2.53) \\ 
  & $\hat{F}_{lcd}$  & 1.11 (1.29) & -1.43 (1.39) & -3.73 (1.67) & -3.43 (1.87) & 0.46 (1.2) & -5.63 (0.24) & -12.78 (0.88) & -11.14 (1.69) & -0.33 (2.02) & 6.4 (0.65) \\ 
   \hline
\end{tabular}
}
\caption{The estimated bias values are reported in units of $ \times 10^2$, and the numbers in parentheses indicate the estimated standard deviations, also in units of $ \times 10^2$. Log-normal($m$, $s$) denotes the truncated log-normal distribution on $[0, 10]$ with mean $m$ and standard deviation $s$ in the log scale.}
\label{table:lognormal_1_2}
\end{table}

\subsection{Estimation performance}
Now, we compare the finite-sample performance of $\hat{F}_{lc}$ against $\hat{F}_{un}$ and $\hat{F}_{lcd}$ across various underlying distributions. In each simulation setting, we consider $1000$ Monte Carlo datasets with a sample size of $250$, $500$ or $1000$. We assess the estimators via the empirical bias and standard deviation of the estimates at the $0.1,0.3,0.5,0.7,0.9$ quantile levels of the underlying distribution function. We consider case 2 interval censoring in this section and provide additional simulation results for current status data in Appendix \ref{app:current_status_sim}.

In Table \ref{table:weibull_1_2},  we consider two event time distributions with a log-concave density: (i) a truncated standard exponential distribution on $[0, 2]$, or (ii) a Weibull distribution with a shape parameter of $2$ and a scale parameter of $1$ truncated on $[0, 2]$. The censoring times are generated as $C_1 \sim \text{Unif}(0, 1)$ and $C_2 \sim \text{Unif}(C_1, 2)$ in both settings. For the Weibull distribution, the empirical biases for $\hat{F}_{lc}$ and $\hat{F}_{lcd}$ are similar, consistent with the observations in Figure \ref{fig:weibull}. For the exponential distribution, we have a log-linear density function, so this corresponds to a boundary case for the log-concave density assumption. Under this scenario, we observe that $\hat{F}_{lc}$ has a smaller bias than $\hat{F}_{lcd}$. Comparing $\hat{F}_{lc}$ with $\hat{F}_{un}$, we see a significant improvement in terms of both bias and standard deviation. 

In Table \ref{table:loglogistic_0.5_1}, we consider two settings where the event times are generated from a truncated log-logistic distribution on $[0, 10]$ with a shape parameter of $0.5$ or $1$, respectively. The censoring times $C_1$ and $C_2$ are generated from Unif$(0, 1)$ and Unif$(C_1, 10)$ in both settings. Both distribution functions are log-concave but do not have log-concave densities. In these cases, $\hat{F}_{lcd}$ suffers from a large bias, while $\hat{F}_{lc}$ outperforms $\hat{F}_{un}$ in terms of smaller empirical bias and standard deviation. The relatively large bias at the $0.1$ quantile level for both $\hat{F}_{lc}$ and $\hat{F}_{un}$ in the log-logistic distribution with a shape parameter of $0.5$ is because the underlying distribution function increases sharply near $0$, while the censoring mechanism leads to limited information in that region.

Table \ref{table:lognormal_1_2} examines two settings where the event times are generated from a truncated log-normal distribution on $[0, 10]$. If $X \sim \text{Log-normal}(m, s)$, then $\log X$ has a normal distribution with mean $m$ and standard deviation $s$. We consider the cases where $(m, s) = (0, 1)$ and $(0, 2)$. The censoring times $C_1$ and $C_2$ are generated from Unif$(0, 1)$ and Unif$(C_1, 10)$ in both settings. Similar to the log-logistic case, a lognormal distribution has a log-concave distribution function but not a log-concave density. As expected, $\hat{F}_{lc}$ outperforms $\hat{F}_{un}$ in terms of smaller empirical bias and standard deviation. Additionally, we observe that the bias of $\hat{F}_{lcd}$ remains constant as the sample size increases, and its magnitude depends on the deviation from log-concavity of the underlying density.

In summary, $\hat{F}_{lc}$ generally exhibits smaller empirical bias and standard deviation than $\hat{F}_{un}$. When the density is log-concave, $\hat{F}_{lcd}$ has a lower standard deviation but may suffer from a higher bias than $\hat{F}_{lc}$. When only the distribution function but not its density is log-concave, $\hat{F}_{lcd}$ can experience a substantial bias. Additional simulation results for current status data with unbounded underlying distributions and discrete censoring times are provided in the supplementary materials.

\section{Real data applications}\label{sect:real_data}

\subsection{Hepatitis A data}
\cite{keiding1991age} considered a cross-sectional study on Hepatitis A virus in Bulgaria, where blood samples were collected from 850 individuals, ranging in age from 1 to 86 years, including school children and blood donors, to determine the presence of Hepatitis A immunity. The dataset is available in the R package \verb|curstatCI|  \citep{groeneboom2017curstatCI} as \verb|hepatitisA|. It contains information about the age of participants at the time of study and whether they are seropositive for Hepatitis A. Since Hepatitis A infection is assumed to confer lifelong immunity, the data represents the immunization status of each individual. However, we only know whether the individual is immunized at the time of study, but not the exact time of infection. This is an example of current status data, a special case of interval censored data. In Figure \ref{fig:real_data_A}, we observe that both the unconstrained MLE and its logarithm appear to have a concave shape. We applied the log-concave MLE to this data, with the results shown in blue in Figure \ref{fig:real_data_A}. The log-concave MLE seems to fit the data well, and closely resembles the shape of the unconstrained MLE while offering a smoother estimate.

\begin{figure}[ht]
    \captionsetup{font=footnotesize}
    \centering
    \includegraphics[width=1\linewidth]{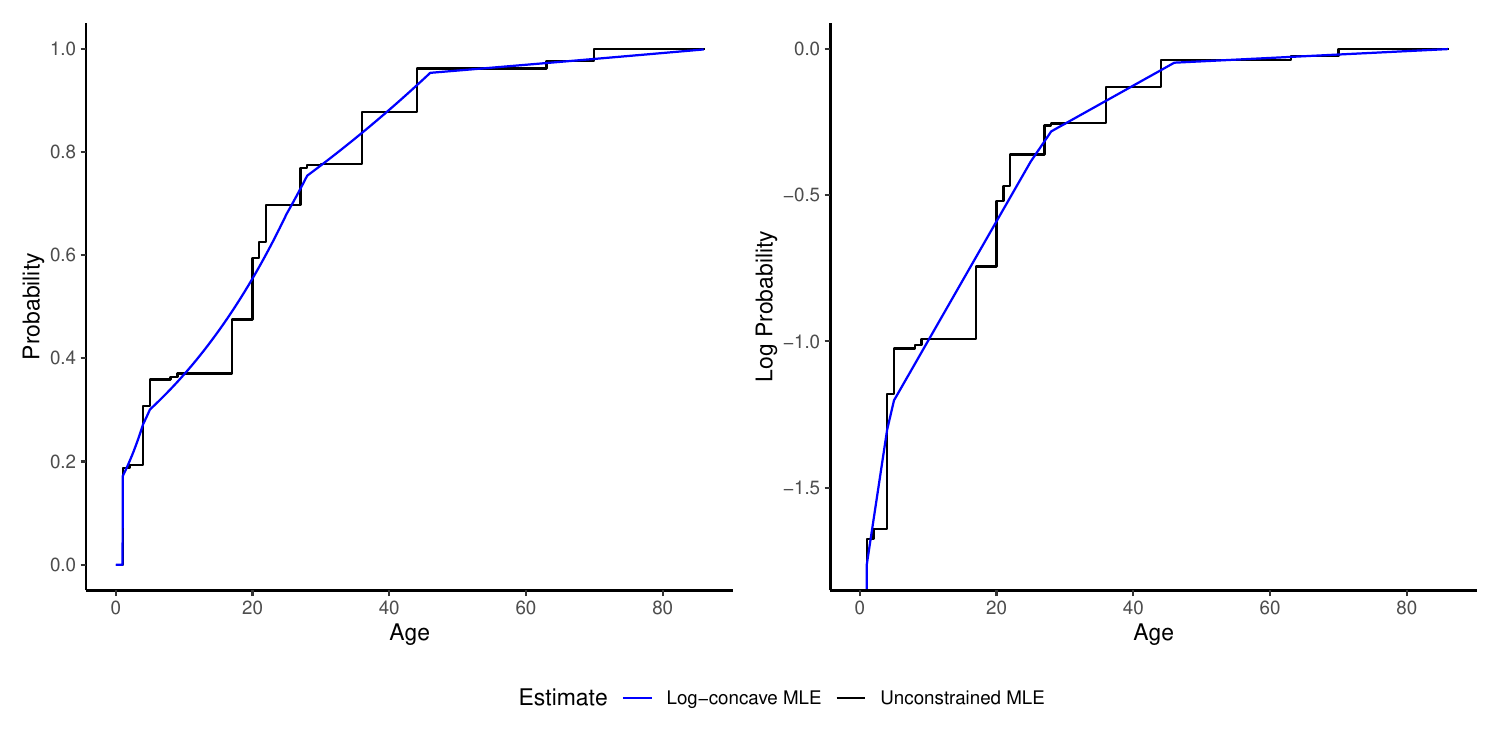}
    \caption{Hepatitis A in Bulgaria: Estimated distribution functions for the age of occurrence (left) and the estimated logarithmic distribution functions (right).}
    \label{fig:real_data_A}
\end{figure}

\subsection{Breast cosmesis data}
The breast cosmesis data from \cite{finkelstein1985semiparametric} has been widely utilized to illustrate new methods for modeling interval-censored data. The study aimed to compare the time to cosmetic deterioration between two groups of early breast cancer patients: 46 patients who received radiotherapy only and 48 patients who received both radiotherapy and chemotherapy. Patients were initially seen at clinic visits every 4 to 6 months, with the interval between visits increasing after completion of primary irradiation treatment.

The study's endpoint was cosmetic deterioration, defined by the appearance of breast retraction. Of the 94 women in the study, 56 experienced deterioration and were interval-censored between the clinic visits, while the remaining 38, who did not experience deterioration, were right-censored at their last clinic visit. The patients in this study were treated at the Joint Center for Radiation Therapy in Boston between 1976 and 1980. For more details, see \cite{beadle1984cosmetic} and  \cite{beadle1984effect}. 

Figures \ref{fig:real_data_rad} and \ref{fig:real_data_radchem} show the unconstrained MLE and log-concave MLE for the radiotherapy only group and the radiotherapy plus chemotherapy group, respectively. In both Figures \ref{fig:real_data_rad} and \ref{fig:real_data_radchem}, we observe that the logarithm of the unconstrained MLEs appear to be concave. Our log-concave MLE fits well for both groups of patients as it closely resembles the shape of the unconstrained MLE while providing a smoother estimate of the underlying distribution function. It is noteworthy that the unconstrained MLE of the distribution function in Figure \ref{fig:real_data_radchem} seems to be first convex then concave, so directly applying a concave-constrained estimator could be problematic.
\begin{figure}[ht]
    \captionsetup{font=footnotesize}
    \centering
    \includegraphics[width=1\linewidth]{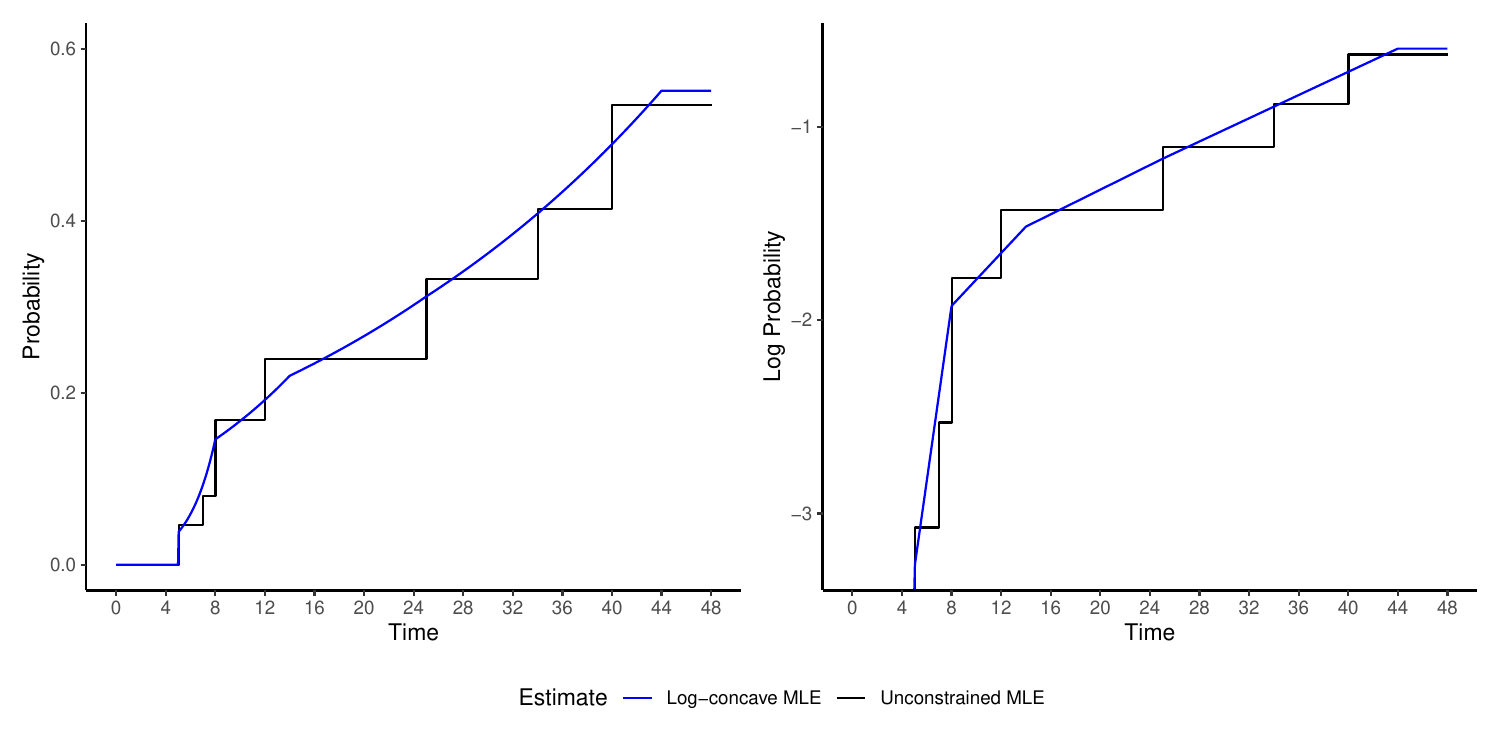}
    \caption{Radiotherapy only group: Estimated distribution functions of the time of deterioration (left) and the estimated logarithmic distribution functions (right).}
    \label{fig:real_data_rad}
\end{figure}

\begin{figure}[ht]
    \captionsetup{font=footnotesize}
    \centering
    \includegraphics[width=1\linewidth]{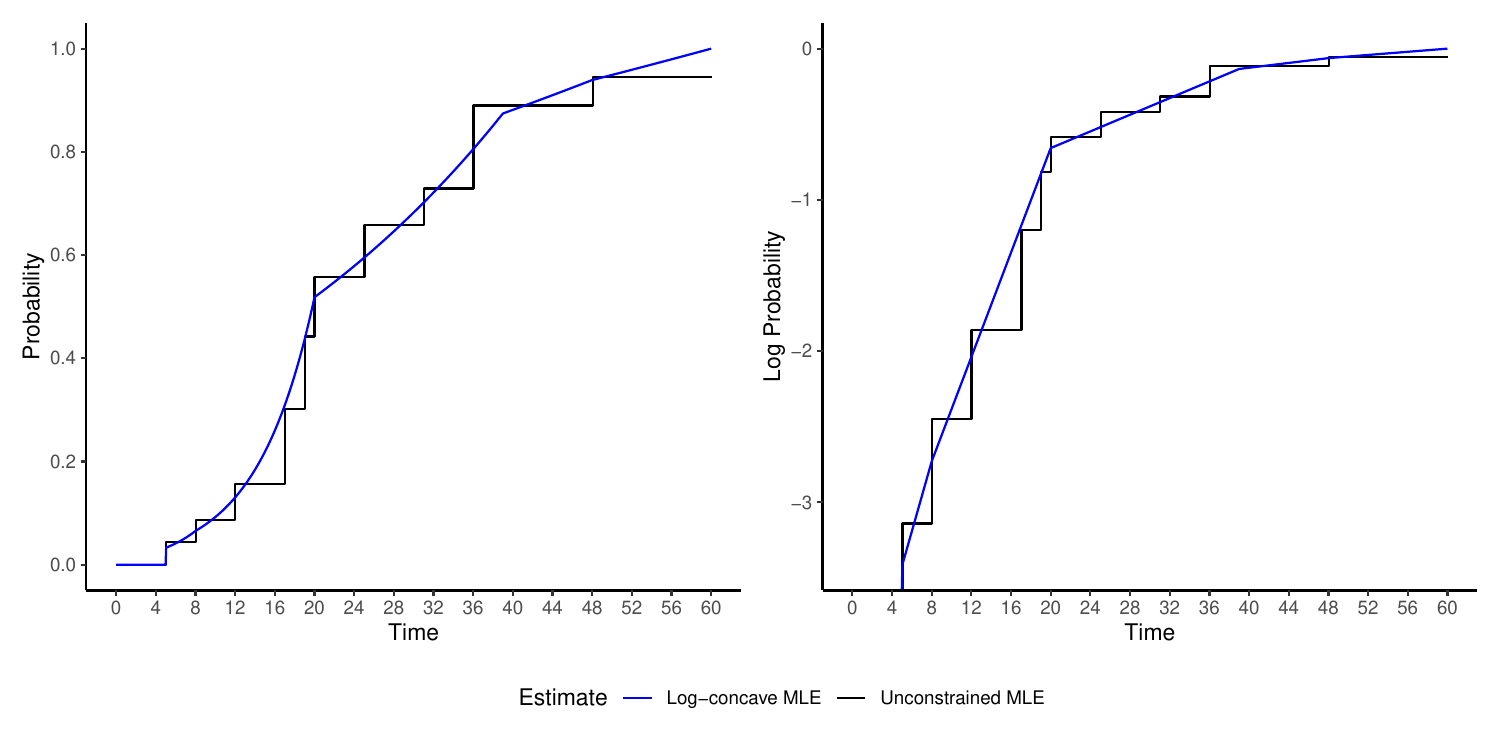}
    \caption{Radiotherapy plus chemotherapy group: Estimated distribution functions of the time of deterioration (left) and the estimated logarithmic distribution functions (right).}
    \label{fig:real_data_radchem}
\end{figure}

\section{Discussion}\label{sect:discussion}
In this article, we consider the nonparametric estimation on mixed-case interval-censored data for a log-concave distribution function. This log-concavity assumption offers greater flexibility than the log-concavity of the density function or the concavity of the distribution function. Simulation studies demonstrate that this shape-constrained estimator significantly outperforms the unconstrained estimator when the underlying distribution function is log-concave.

A future research direction is to develop a goodness-of-fit test for the log-concavity of $F_0$ based on interval-censored data. Currently, a graphical check that can be performed is to overlay the unconstrained MLE and the log-concave MLE as illustrated in the real data applications.
Another important direction for future work is to develop an inference procedure based on the log-concave MLE, for example, by first deriving its asymptotic distribution. In general, establishing the asymptotic distribution of shape-constrained estimators is a challenging task and typically needs to be addressed on a case-by-case basis. 

\appendix
\section{Appendix: Proofs for Section \ref{sect:exist_unique_consistency}}\label{sect:proof}
   
\begin{proof}[Proof of Lemma \ref{lemma:log_like_concave}]
To show that $l$ is concave, it suffices to show that the function $g(x, y) := \log(e^x - e^y)$, for $x \geq y$, is concave, as $l$ is a weighted sum of functions of this type. By straightforward calculation, the Hessian matrix of $g$ is
\begin{equation*}
    H_g = \frac{e^x e^y}{(e^x - e^y)^2} \left(
\begin{array}{cc}
     -1 & 1 \\
     1& -1 
\end{array}
    \right).
\end{equation*}
For any $v = (v_1, v_2)^\top \in \mathbb{R}^2$, $v^\top H_g v = -\frac{e^x e^y}{(e^x - e^y)^2}(v_1 - v_2)^2 \leq 0$. Thus, $H_g$ is negative semi-definite, implying that $g$ is concave. 

To see that $g$ is not strictly concave, consider $(x_1,y_1) = (0, -1)$ and $(x_2, y_2) = (-1, -2)$. Then, for any $t \in [0, 1]$, it is straightforward to see that
\begin{align*}
    tg(x_1, y_1) + (1-t) g(x_2,y_2)&=
    t \log (e^0 -  e^{-1}) + (1-t) \log (e^{-1} - e^{-2}) \\
    &= \log (e^{t-1}- e^{t-2}) \\
    &=g(tx_1+(1-t)x_2, ty_1+(1-t)y_2).
\end{align*}
\end{proof}

    \begin{proof}[Proof of Theorem \ref{thm:unique_existence}]
We first show the existence of a maximizer. Recall the definition of $\mathcal{L}$ in \eqref{eq:mathcal_L}. After setting $\phi_i = -\infty$ for $i=1,\ldots,s^*$, where $s^* = |\{L_i: i \in \mathcal{L}\}|$. The log-likelihood includes at least the term $\log(e^{\phi_{s^*+1}}) = \phi_{s^*+1}$. Given the constraints $\phi_{s^*+1} \leq \ldots \leq \phi_m \leq 0$, it follows that $\phi_{s^*+1} \rightarrow -\infty$ if $\sqrt{\sum_{s^*+1}^m \phi_i^2} \rightarrow \infty$. Therefore, $l(\phi) \rightarrow -\infty$ as $\sqrt{\sum_{s^*+1}^m \phi_i^2}  \rightarrow \infty$. Moreover, since $l$ is continuous and the constraints given by \eqref{eq:monotone_constraint} and \eqref{eq:log_concave_constraint} form a closed set, it follows that $l$ has a maximizer.

The proof of uniqueness is similar to that in \cite{dumbgen2006estimating} for the concave-constrained MLE. Because of our parameterization in terms of the log distribution function, some elements of the maximizer take value of negative infinity (corresponding to the maximizer in terms of the distribution function taking value of $0$). 
    Suppose that $\hat{\phi}$ and $\hat{\phi} + D$ are two different maximizers of $l$ over $\mathcal{M} \cap \mathcal{LC}$, where $D \in \mathbb{R}^m$. 
    Without loss of generality, we can set $D_i = 0$ for $i=1,\ldots,s^*$ because $\hat{\phi}_1 = \ldots = \hat{\phi}_{s^*} = -\infty$.    
    Let $\hat{\phi}^M_i = \hat{\phi}_i$ for $i > s^*$ and $\hat{\phi}^M_i = M$ for $i \leq s^*$, with $M < \hat{\phi}_{s^*+1}$. Then, we can see that both $\hat{\phi}^M$ and $\hat{\phi}^M + D$ are maximizers of $l(\phi)$ subject to $\phi_i = M$ for $i=1,\ldots,s^*$, $\Delta \phi_i/\Delta \tau_i \geq \Delta \phi_{i+1}/\Delta \tau_{i+1}$ for $i= s^*+2,\ldots,m-1$, and $\phi_i \leq \phi_{i+1}$ for $i=1,\ldots,m-1$. To show the uniqueness of $\hat{\phi}$, it suffices to show the uniqueness of $\hat{\phi}^M$.

  Let $\tilde{\phi} := (1-t)\hat{\phi}^M + t(\hat{\phi}^M+D) = \hat{\phi}^M + tD$. Note that $l(\hat{\phi}^M + t D)$ is constant for all $t \in [0, 1]$ as $l$ is concave and both $\hat{\phi}^M$ and $\hat{\phi}^M+ D$ are maximizers. Thus,
\begin{equation*}
    \frac{d^2}{dt^2} l( \hat{\phi}^M +t D) = 0.
\end{equation*}
By straightforward calculation, we obtain
\begin{equation}\label{eq:dt2}
    \frac{d^2}{dt^2} l( \hat{\phi}^M + t D) = -\sum^n_{i=1} w_i \frac{e^{\hat{\phi}_{r(i)}^M + \hat{\phi}_{l(i)}^M + t(D_{r(i)}+ D_{l(i)}) }}{ (e^{\hat{\phi}_{r(i)}^M + tD_{r(i)}} -  
    e^{\hat{\phi}_{l(i)}^M + tD_{l(i)}})^2} \cdot (D_{r(i)} - D_{l(i)})^2.
\end{equation}
Since $\hat{\phi}_{r(i)}^M + \hat{\phi}_{l(i)}^M + t(D_{r(i)}+ D_{l(i)})$ is finite and the denominator $e^{\hat{\phi}_{r(i)} + tD_{r(i)}} - e^{\hat{\phi}_{l(i)} + tD_{l(i)}}$ cannot be zero because otherwise $\tilde{\phi}$ cannot be a maximizer as the log-likelihood becomes negative infinity, we have
\begin{equation}\label{eq:D_unique}
    D_{r(i)} - D_{l(i)} = 0, \quad \text{for } i=1,\ldots,n.
\end{equation}
Let $a$ be the smallest index in $\{1,\ldots,m\}$ such that $D_a \neq 0$. Then, let $b$ be the largest index in $\{a,\ldots,m\}$ such that $\Delta \hat{\phi}_j^M /\Delta \tau_j$ and $\Delta \tilde{\phi}_j /\Delta \tau_j$ are constant in $j\in \{a,\ldots,b\}$. This implies that $\Delta D_j/ \Delta \tau_j$ is constant in $j\in\{a,\ldots,b\}$. Furthermore, because the slope changes at $\tau_b$, there must be some $i_0$ such that $r(i_0) = b$. If not, $\hat{\phi}_b^M$ only appears in $\log (e^{\hat{\phi}_{r(i)}^M} -e^{\hat{\phi}_b^M})$ and a strictly greater log-likelihood can be obtained by changing the value of $\hat{\phi}_b^M$ to be some interpolated value of $\hat{\phi}_{b+1}^M$ and $\hat{\phi}_{b-1}^M$. This contradicts the fact that $\hat{\phi}^M$ is a maximizer. Recall that $\Delta D_j/ \Delta \tau_j$ is constant in $j \in \{a,\ldots,b\}$. If $l(i_0) \geq a$, we have
\begin{equation}\label{eq:D_unqiue_contrad1}
    \frac{D_{r(i_0)} - D_{l(i_0)}}{\tau_{r(i_0)} - \tau_{l(i_0)}} = \frac{\Delta D_a}{\Delta \tau_a} \neq 0.
\end{equation}
If $l(i_0) < a$, we have $D_{l(i_0)} = D_{a-1} = 0$ and
\begin{equation}\label{eq:D_unqiue_contrad2}
    \frac{D_{r(i_0)} - D_{l(i_0)}}{\tau_{r(i_0)} - \tau_{a-1}} = \frac{\Delta D_a}{\Delta \tau_a} \neq 0.
\end{equation}
Since \eqref{eq:D_unqiue_contrad1} or \eqref{eq:D_unqiue_contrad2} contradicts with (\ref{eq:D_unique}), $D$ must be $0$ and so $\hat{\phi}^M$ is unique.

\end{proof}

\begin{proof}[Proof of Theorem \ref{thm:consistency}]
    We apply Theorem 3 in \cite{dumbgen2006estimating} with their $\mathcal{F}$ being $\mathcal{F}_{lc}$. Their Condition (C.1) is satisfied by our assumption. Then, equation \eqref{eq:consistency} follows the same argument that leads to equation (16) in \cite{dumbgen2006estimating}.
\end{proof}

\section{Appendix: Additional simulation results}\label{app:current_status_sim}
In this section, we present additional simulation results of our estimator for current status data. The event time follows (untruncated) exponential, Weibull, log-logistic, or log-normal distribution. The censoring time is generated from a standard exponential distribution. As shown in Table \ref{table:additional}, the log-concave MLE consistently outperforms the unconstrained MLE, with smaller biases and standard deviations across different quantile levels.

\begin{table}[H]
\captionsetup{font=footnotesize}
\resizebox{\textwidth}{!}{
\centering
\begin{tabular}{r|c|rrrrr||rrrrr}
  \hline
  & & \multicolumn{5}{c||}{Exp(1) } & \multicolumn{5}{c}{Weibull(2, 1)} \\ \hline
   $N$ & Est.& 0.1 & 0.3 & 0.5 & 0.7 & 0.9 & 0.1 & 0.3 & 0.5 & 0.7 & 0.9 \\ \hline
250 & $\hat{F}_{lc}$ & -0.84 (6.82) & 0.02 (5.65) & -0.22 (5.79) & 0.32 (5.92) & 1.05 (5.1) & 0.57 (5.2) & 0.16 (5.86) & -0.02 (7.17) & 0.33 (7.2) & 0.98 (5.5) \\ 
 & $\hat{F}_{un}$ & -2.57 (6.2) & -0.85 (8.21) & -0.71 (8.23) & 0.68 (8.13) & 1.99 (6.53) & -1.34 (5.59) & -1.17 (9.14) & -0.48 (10.71) & 0.96 (10.14) & 1.72 (7.19) \\ 
  500 &$\hat{F}_{lc}$ & 0.01 (4.68) & -0.1 (4.35) & 0.1 (4.48) & 0.09 (4.22) & 0.69 (3.75) & 0.47 (3.28) & -0.11 (4.33) & 0.15 (5.38) & 0.34 (5.36) & 0.35 (4.14) \\ 
  &$\hat{F}_{un}$ & -1.59 (4.94) & -0.66 (6.23) & -0.03 (6.59) & 0.14 (6.17) & 1.46 (4.94) & -0.77 (4.42) & -0.78 (6.95) & -0.08 (8.2) & 0.55 (7.77) & 1.02 (5.59) \\ 
  1000 &$\hat{F}_{lc}$ & 0.32 (2.96) & 0.01 (3.32) & 0.02 (3.36) & -0.03 (3.35) & 0.24 (2.77) & 0.13 (2.23) & -0.1 (3.49) & 0.02 (3.98) & -0.09 (3.74) & 0.04 (2.87) \\ 
  &$\hat{F}_{un}$ & -0.71 (3.84) & -0.22 (4.94) & -0.12 (5.08) & 0.1 (5.03) & 0.84 (3.68) & -0.49 (3.47) & -0.38 (5.77) & 0.02 (6.53) & -0.02 (6.05) & 0.58 (4.09) \\  \hline
  & & \multicolumn{5}{c||}{Log-logistic(5, 1) } & \multicolumn{5}{c}{Log-normal(0, 0.5)} \\ \hline
   $N$ & Est.& 0.1 & 0.3 & 0.5 & 0.7 & 0.9 & 0.1 & 0.3 & 0.5 & 0.7 & 0.9 \\ \hline
250 &$\hat{F}_{lc}$ & -0.13 (6.52) & 0.36 (7.69) & -0.41 (8.75) & -0.31 (8.36) & 0.7 (5.79) & -0.09 (6.51) & 0.4 (6.69) & -0.29 (7.51) & -0.25 (7.13) & 0.97 (5.31) \\ 
  &$\hat{F}_{un}$ & -2.06 (6.65) & -1.84 (11.4) & -0.91 (12.6) & 0.2 (11.7) & 1.47 (7.32) & -1.9 (6.5) & -1.06 (10.03) & -0.74 (10.7) & -0.12 (9.92) & 1.74 (6.72) \\ 
  500 &$\hat{F}_{lc}$ & 0.61 (3.97) & 0.02 (5.72) & -0.15 (6.81) & 0.03 (6.57) & 0.47 (4.17) & 0.49 (4) & 0.04 (5.28) & 0.41 (6) & 0.64 (5.45) & 0.34 (3.93) \\ 
  &$\hat{F}_{un}$ & -1.04 (5.12) & -0.82 (9.01) & -0.36 (10.15) & 0.21 (9.46) & 1.1 (5.41) & -1.11 (5.01) & -0.74 (8) & 0.28 (8.7) & 0.84 (7.84) & 0.85 (5.04) \\ 
  1000 &$\hat{F}_{lc}$ & 0.42 (2.5) & -0.13 (4.21) & -0.01 (5.25) & -0.09 (4.85) & 0.14 (3.05) & 0.15 (2.52) & -0.05 (4.06) & 0.03 (4.4) & 0.04 (4.18) & 0.42 (3.06) \\ 
  &$\hat{F}_{un}$ & -0.55 (4.05) & -0.39 (6.88) & -0.28 (7.88) & -0.09 (7.15) & 0.5 (4.08) & -0.77 (3.83) & -0.49 (6.25) & -0.23 (6.78) & 0.3 (6.06) & 0.95 (3.91) \\ 
   \hline
\end{tabular}
}
\caption{The estimated bias values are reported in units of $ \times 10^2$, and the numbers in parentheses indicate the estimated standard deviations, also in units of $ \times 10^2$. The censoring times follow a standard exponential distribution.}
\label{table:additional}
\end{table}

In real data applications, we may encounter situations where the inspection times are not precisely known but rounded to the closest day or month, for example. Table \ref{table:additional2} presents the simulation results under a similar setting as in Table \ref{table:additional} except that the censoring times are only observed up to a precision of $0.1$. This also results in some possible ties in the observations and thus the weights $w_i$ can be greater than $1$. As shown, the empirical bias in the unconstrained MLE can persist in this situation, whereas the log-concave MLE exhibits considerably smaller bias due to its smoothing effect, even when the censoring times are not continuous.

\begin{table}[H]
\captionsetup{font=footnotesize}
\resizebox{\textwidth}{!}{
\centering
\begin{tabular}{r|c|rrrrr||rrrrr}
  \hline
  & & \multicolumn{5}{c||}{Exp(1) } & \multicolumn{5}{c}{Weibull(2, 1)} \\ \hline
   $N$ & Est.& 0.1 & 0.3 & 0.5 & 0.7 & 0.9 & 0.1 & 0.3 & 0.5 & 0.7 & 0.9 \\ \hline
250 & $\hat{F}_{lc}$ & -1.25 (5.24) & -0.44 (5.44) & -0.3 (5.67) & 0.18 (5.73) & 1.24 (4.84) & -0.95 (4.38) & 0.05 (6.17) & 0.23 (7.22) & 0.29 (6.88) & 0.82 (5.27) \\ 
&  $\hat{F}_{un}$ & -0.86 (5.63) & -4.83 (7.45) & -5.4 (7.94) & -0.08 (8.13) & 1.7 (6.49) & -2.07 (4.93) & -8.62 (7.81) & -2.43 (10.13) & -6.67 (10.47) & 0.28 (6.91) \\ 
500 &  $\hat{F}_{lc}$ & -1.01 (3.63) & -0.32 (3.99) & 0.02 (4.35) & -0.07 (4.51) & 0.57 (3.82) & -0.31 (3.09) & 0.17 (4.64) & -0.13 (5.14) & 0.21 (5.4) & 0.46 (4.07) \\ 
 & $\hat{F}_{un}$ & -0.81 (3.9) & -4.22 (5.78) & -5.59 (6.37) & -0.31 (6.35) & 0.85 (5.17) & -1.52 (3.97) & -7.8 (6.24) & -2.77 (7.85) & -6.86 (7.9) & -0.19 (5.52) \\ 
 1000 & $\hat{F}_{lc}$ & -0.62 (2.76) & -0.51 (3.06) & 0.21 (3.52) & 0.29 (3.27) & 0.33 (2.62) & -0.23 (2.13) & -0.06 (3.4) & -0.25 (3.8) & 0.02 (4.07) & 0.01 (2.88) \\ 
 & $\hat{F}_{un}$ & -0.62 (3) & -4.54 (4.49) & -5.12 (5.12) & 0.04 (4.91) & 0.42 (3.61) & -1.43 (2.93) & -8.12 (4.64) & -2.76 (5.96) & -6.78 (6.26) & -0.54 (4.17) \\ 
   \hline
  & & \multicolumn{5}{c||}{Log-logistic(5, 1) } & \multicolumn{5}{c}{Log-normal(0, 0.5)} \\ \hline
   $N$ & Est.& 0.1 & 0.3 & 0.5 & 0.7 & 0.9 & 0.1 & 0.3 & 0.5 & 0.7 & 0.9 \\ \hline
250 & $\hat{F}_{lc}$ & -1.02 (5.66) & -0.3 (7.4) & -0.05 (9.08) & -0.62 (8.23) & 0.72 (5.5) & -0.73 (5.06) & -0.07 (6.74) & 0.31 (7.87) & 0.18 (7.44) & 1.11 (5.56) \\ 
 & $\hat{F}_{un}$ & -3.08 (5.42) & -5.56 (9.91) & -0.15 (12.24) & -8.73 (11.56) & -1.26 (7.95) & -2.11 (5.3) & -6.26 (8.94) & -0.05 (10.54) & -5.96 (10.23) & -1.05 (7.67) \\ 
500 &  $\hat{F}_{lc}$ & -0.29 (3.63) & -0.64 (5.33) & -0.34 (6.83) & -0.4 (6.32) & 0.31 (4.2) & -0.33 (3.4) & -0.18 (5.05) & 0.13 (5.96) & 0.27 (5.58) & 0.49 (4.06) \\ 
 &  $\hat{F}_{un}$ & -2.86 (4.08) & -5.51 (7.74) & -0.58 (9.65) & -8.7 (9.45) & -1.62 (5.79) & -1.89 (4.08) & -6.13 (6.87) & -0.21 (8.25) & -6.01 (8.08) & -1.77 (5.55) \\ 
 1000 & $\hat{F}_{lc}$ & -0.13 (2.39) & -0.42 (3.92) & 0.15 (5.36) & -0.29 (4.6) & -0.05 (2.95) & -0.24 (2.54) & -0.27 (3.7) & -0.03 (4.33) & -0.16 (4.11) & 0.16 (3.08) \\ 
 & $\hat{F}_{un}$ & -2.92 (3.2) & -5.56 (5.83) & 0.14 (7.28) & -8.48 (7.07) & -1.7 (4.49) & -1.78 (3.33) & -6.1 (5.26) & -0.06 (6.47) & -6.12 (6.11) & -2.12 (4.33) \\ 
   \hline
\end{tabular}
}
\caption{The estimated bias values are reported in units of $ \times 10^2$, and the numbers in parentheses indicate the estimated standard deviations, also in units of $ \times 10^2$. The censoring times follow a standard exponential distribution rounded to the nearest $0.1$.}
\label{table:additional2}
\end{table}
\bibliographystyle{chicago}

\bibliography{main}

\begin{thebibliography}{}

\bibitem[\protect\citeauthoryear{Anderson-Bergman}{Anderson-Bergman}{2014}]{Anderson-Bergman2014logconPH}
Anderson-Bergman, C. (2014).
\newblock {\em logconPH: CoxPH Model with Log Concave Baseline Distribution}.
\newblock R package version 1.5.

\bibitem[\protect\citeauthoryear{Anderson-Bergman}{Anderson-Bergman}{2017}]{anderson2017icenreg}
Anderson-Bergman, C. (2017).
\newblock icenreg: regression models for interval censored data in r.
\newblock {\em Journal of Statistical Software\/}~{\em 81}, 1--23.

\bibitem[\protect\citeauthoryear{Anderson-Bergman and Yu}{Anderson-Bergman and
  Yu}{2016}]{anderson2016computing}
Anderson-Bergman, C. and Y.~Yu (2016).
\newblock Computing the log concave npmle for interval censored data.
\newblock {\em Statistics and Computing\/}~{\em 26}, 813--826.

\bibitem[\protect\citeauthoryear{Bagnoli and Bergstrom}{Bagnoli and
  Bergstrom}{2005}]{bagnoli2005log}
Bagnoli, M. and T.~Bergstrom (2005).
\newblock Log-concave probability and its applications.
\newblock {\em Economic Theory\/}~{\em 26}, 445--469.

\bibitem[\protect\citeauthoryear{Beadle, Come, Henderson, Silver, Hellman, and
  Harris}{Beadle et~al.}{1984}]{beadle1984effect}
Beadle, G.~F., S.~Come, I.~C. Henderson, B.~Silver, S.~Hellman, and J.~R.
  Harris (1984).
\newblock The effect of adjuvant chemotherapy on the cosmetic results after
  primary radiation treatment for early stage breast cancer.
\newblock {\em International Journal of Radiation Oncology* Biology*
  Physics\/}~{\em 10\/}(11), 2131--2137.

\bibitem[\protect\citeauthoryear{Beadle, Silver, Botnick, Hellman, and
  Harris}{Beadle et~al.}{1984}]{beadle1984cosmetic}
Beadle, G.~F., B.~Silver, L.~Botnick, S.~Hellman, and J.~R. Harris (1984).
\newblock Cosmetic results following primary radiation therapy for early breast
  cancer.
\newblock {\em Cancer\/}~{\em 54\/}(12), 2911--2918.

\bibitem[\protect\citeauthoryear{Busing}{Busing}{2022}]{busing2022monotone}
Busing, F.~M. (2022).
\newblock Monotone regression: A simple and fast o (n) pava implementation.
\newblock {\em Journal of Statistical Software\/}~{\em 102}, 1--25.

\bibitem[\protect\citeauthoryear{D{\"u}mbgen, Freitag, and
  Jongbloed}{D{\"u}mbgen et~al.}{2004}]{dumbgen2004consistency}
D{\"u}mbgen, L., S.~Freitag, and G.~Jongbloed (2004).
\newblock Consistency of concave regression with an application to
  current-status data.
\newblock {\em Mathematical methods of statistics\/}~{\em 13}, 69--81.

\bibitem[\protect\citeauthoryear{D{\"u}mbgen, Freitag-Wolf, and
  Jongbloed}{D{\"u}mbgen et~al.}{2006}]{dumbgen2006estimating}
D{\"u}mbgen, L., S.~Freitag-Wolf, and G.~Jongbloed (2006).
\newblock Estimating a unimodal distribution from interval-censored data.
\newblock {\em Journal of the American Statistical Association\/}~{\em
  101\/}(475), 1094--1106.

\bibitem[\protect\citeauthoryear{D{\"u}mbgen, H{\"u}sler, and
  Rufibach}{D{\"u}mbgen et~al.}{2011}]{dumbgen2007active}
D{\"u}mbgen, L., A.~H{\"u}sler, and K.~Rufibach (2011).
\newblock Active set and em algorithms for log-concave densities based on
  complete and censored data.

\bibitem[\protect\citeauthoryear{D{\"u}mbgen, Rufibach, and
  Schuhmacher}{D{\"u}mbgen et~al.}{2014}]{dumbgen2014maximum}
D{\"u}mbgen, L., K.~Rufibach, and D.~Schuhmacher (2014).
\newblock Maximum-likelihood estimation of a log-concave density based on
  censored data.
\newblock {\em Electronic Journal of Statistics\/}~{\em 8}, 1405–1437.

\bibitem[\protect\citeauthoryear{Finkelstein and Wolfe}{Finkelstein and
  Wolfe}{1985}]{finkelstein1985semiparametric}
Finkelstein, D.~M. and R.~A. Wolfe (1985).
\newblock A semiparametric model for regression analysis of interval-censored
  failure time data.
\newblock {\em Biometrics\/}~{\em 41\/}(4), 933--945.

\bibitem[\protect\citeauthoryear{Groeneboom}{Groeneboom}{2017}]{groeneboom2017curstatCI}
Groeneboom, P. (2017).
\newblock {\em curstatCI: Confidence Intervals for the Current Status Model}.
\newblock R package version 0.1.1.

\bibitem[\protect\citeauthoryear{Groeneboom and Jongbloed}{Groeneboom and
  Jongbloed}{2014}]{groeneboom2014nonparametric}
Groeneboom, P. and G.~Jongbloed (2014).
\newblock {\em Nonparametric estimation under shape constraints}, Volume~38.
\newblock Cambridge: Cambridge University Press.

\bibitem[\protect\citeauthoryear{Groeneboom and Wellner}{Groeneboom and
  Wellner}{1992}]{groeneboom1992information}
Groeneboom, P. and J.~A. Wellner (1992).
\newblock {\em Information bounds and nonparametric maximum likelihood
  estimation}, Volume~19.
\newblock Basel: Birkh{\"a}user.

\bibitem[\protect\citeauthoryear{Gruttola and Lagakos}{Gruttola and
  Lagakos}{1989}]{de1989analysis}
Gruttola, V.~D. and S.~W. Lagakos (1989).
\newblock Analysis of doubly-censored survival data, with application to aids.
\newblock {\em Biometrics\/}~{\em 45\/}(1), 1--11.

\bibitem[\protect\citeauthoryear{Jongbloed}{Jongbloed}{1998}]{jongbloed1998iterative}
Jongbloed, G. (1998).
\newblock The iterative convex minorant algorithm for nonparametric estimation.
\newblock {\em Journal of Computational and Graphical Statistics\/}~{\em
  7\/}(3), 310--321.

\bibitem[\protect\citeauthoryear{Keiding}{Keiding}{1991}]{keiding1991age}
Keiding, N. (1991).
\newblock Age-specific incidence and prevalence: a statistical perspective.
\newblock {\em Journal of the Royal Statistical Society Series A: Statistics in
  Society\/}~{\em 154\/}(3), 371--396.

\bibitem[\protect\citeauthoryear{Samworth}{Samworth}{2018}]{samworth2018recent}
Samworth, R.~J. (2018).
\newblock Recent progress in log-concave density estimation.
\newblock {\em Statistical Science\/}~{\em 33\/}(4), 493--509.

\bibitem[\protect\citeauthoryear{Schick and Yu}{Schick and
  Yu}{2000}]{schick2000consistency}
Schick, A. and Q.~Yu (2000).
\newblock Consistency of the gmle with mixed case interval-censored data.
\newblock {\em Scandinavian Journal of Statistics\/}~{\em 27\/}(1), 45--55.

\bibitem[\protect\citeauthoryear{Sengupta}{Sengupta}{2006}]{sengupta2004concave}
Sengupta, D. (2006).
\newblock {\em Concave and Log-Concave Distributions}, Chapter~2.
\newblock New York: John Wiley \& Sons, Ltd.

\bibitem[\protect\citeauthoryear{Sengupta and Nanda}{Sengupta and
  Nanda}{1999}]{sengupta1999log}
Sengupta, D. and A.~K. Nanda (1999).
\newblock Log-concave and concave distributions in reliability.
\newblock {\em Naval Research Logistics (NRL)\/}~{\em 46\/}(4), 419--433.

\bibitem[\protect\citeauthoryear{Sun}{Sun}{2006}]{sun2006statistical}
Sun, J. (2006).
\newblock {\em The statistical analysis of interval-censored failure time
  data}.
\newblock New York: Springer.

\bibitem[\protect\citeauthoryear{Turnbull}{Turnbull}{1976}]{turnbull1976empirical}
Turnbull, B.~W. (1976).
\newblock The empirical distribution function with arbitrarily grouped,
  censored and truncated data.
\newblock {\em Journal of the Royal Statistical Society. Series B
  (Methodological)\/}~{\em 38\/}(3), 290--295.

\bibitem[\protect\citeauthoryear{Van~de Geer}{Van~de
  Geer}{1993}]{van1993hellinger}
Van~de Geer, S. (1993).
\newblock Hellinger-consistency of certain nonparametric maximum likelihood
  estimators.
\newblock {\em The Annals of Statistics\/}~{\em 21\/}(1), 14--44.

\bibitem[\protect\citeauthoryear{Vanobbergen, Martens, Declerck, and
  Lesaffre}{Vanobbergen et~al.}{2000}]{vanobbergen2000tooth}
Vanobbergen, J., L.~Martens, D.~Declerck, and M.~Lesaffre (2000).
\newblock The signal tandmobiel (r) project: a longitudinal intervention oral
  health promotion study in flanders (belgium) baseline and first year results.
\newblock {\em European Journal of Paediatric Dentistry\/}~{\em 2}, 87--96.

\bibitem[\protect\citeauthoryear{Wang}{Wang}{2018}]{wang2018computation}
Wang, Y. (2018).
\newblock Computation of the nonparametric maximum likelihood estimate of a
  univariate log-concave density.
\newblock {\em Wiley Interdisciplinary Reviews: Computational
  Statistics\/}~{\em 11\/}(1), e1452.

\end{thebibliography}
\end{document}